\documentclass[sn-mathphys]{sn-jnl}

\usepackage{amssymb,bbm,amsmath,hyperref,booktabs,multirow,caption}
\renewcommand\rightmark{Optimal portfolio selection in the market model with random coefficients}
\renewcommand\leftmark{Jeong Yin Park}
\hypersetup{ colorlinks=true, linkcolor=blue, filecolor=magenta, urlcolor=cyan }

\jyear{2022}%

\theoremstyle{thmstyleone}%
\newtheorem{theorem}{Theorem}
\newtheorem{proposition}[theorem]{Proposition}%

\theoremstyle{thmstyletwo}%
\newtheorem{remark}{Remark}%

\theoremstyle{thmstylethree}%
\newtheorem{definition}{Definition}%

\begin{document}

\title{Optimal portfolio selection of many players under relative performance criteria in the market model with random coefficients}


\author*[1]{\fnm{Jeong Yin} \sur{Park}}\email{jeongyin.park@kaist.ac.kr}

\affil[1]{\orgdiv{Department of Mathematical Sciences}, \orgname{Korea Advanced Institute of Science and Technology (KAIST)}, \orgaddress{\city{Daejeon}, \postcode{34141}, \country{Republic of Korea}}}


\abstract{We study the optimal portfolio selection problem under relative performance criteria in the market model with random coefficients from the perspective of many players game theory. We consider five random coefficients which consist of three market parameters which are used in the risky asset price modeling and two preference parameters which are related to risk attitude and impact of relative performance. We focus on two cases; either all agents have Constant Absolute Risk Aversion (CARA) risk preferences or all agents have Constant Relative Risk Aversion (CRRA) risk preferences for their investment optimization problem. For each case, we show that the forward Nash equilibrium and the mean field equilibrium exist for the $n$-agent game and the corresponding mean field stochastic optimal control problem, respectively. To extend the $n$-agent game to the continuum of players game, we introduce a measure dependent forward relative performance process and apply an optimization over controlled dynamics of McKean-Vlasov type. We conclude that our optimal portfolio formulas extend the corresponding results of the market model with constant coefficients.
}

\keywords{Optimal investment, Forward relative performance process, Relative performance, Mean field stochastic optimal control, Common noise Merton problem}

\pacs[JEL Classification]{G11, C73}


\maketitle

\section{Introduction}\label{sec1}
\indent We study the optimal investment problem under relative performance criteria in the market model with random coefficients. We derive the explicit solutions for both $n$-player game and mean field type control problem. Our market model can be informally outlined, with full details given in Sect. \ref{sec3}, as follows: The agents can only trade in a common riskless asset and a specific risky asset that is influenced by a common market signal. Each agent measures its relative performance, taking into account competition among agents, and using the difference between their own wealth and the average wealth of all agents as a benchmark.

Throughout the study, we consider asset specialization, which represents the case that each agent can trade between their individual assigned risky asset and a common riskless asset. Asset specialization can be observed for a variety of reasons. Typically, agents are familiar with a certain field, and they want to reduce the cost to improve their knowledge of new stocks, as well as whatever trading costs and constraints exist in the market. The evidence for the above reasons is well explained in terms of empirical research(\cite{CM1,FK1,KS1}) and theoritical research(\cite{AGW1,VV1}).

The classical backward approach to finding an optimal portfolio that maximizes the expected utility has limitations in that it cannot be applied when the utility changes over time or when the maturity is not fixed. Therefore, in order to compensate for the shortcomings of the existing method, Musiela and Zariphopoulou \cite{MZ1} introduce a random field called the forward performance process. See also \cite{MZ2,Z1}. We focus on a class of forward performance processes that is differentiable in time. In other words, its Itô dynamics has zero volatility.

Mean field game theory was first introduced in \cite{HM1,LL1} and it is used in various fields of economics and financial mathematics. See also \cite{LL3}. For example, it is applied in economic growth, optimal execution, systemic risk, and large population behavior models. In previous studies, there have been few explicitly solvable mean field game models in the presence of common noise. See \cite{CF1,LS1,LZ1,RP1,RP3}. We add two examples of explicitly solvable mean field type models. To obtain a rigorous analysis of special nonlinear partial differential equations, McKean-Vlasov type stochastic differential equation was first introduced by McKean in \cite{M1,M2}. Lasry and Lions interpret the mean field games as a criticallity for some McKean-Vlasov control problem in \cite{LL1,LL2,LL3}. Carmona et al. \cite{CDL1} thoroughly examine the differences between mean field games and McKean-Vlasov problems in full detail. In short, the difference between the two methods is which happens first: optimization or performing passaging to the limit. If we optimize first, the asymptotic problem is called a mean-field game. On the other hand, if passaging to the limit is performed first, it is refered to as an optimization problem over controlled dynamics of McKean-Vlasov type or a mean-field stochastic optimal control.

Various previous studies have been conducted from the perspective of many player games to select the optimal portfolio under the relative performance concerns. Lacker and Zariphopoulou \cite{LZ1} obtain the optimal portfolio using the backward approach, which maximizes the expected utility. They consider two typical utilities, exponential Constant Absolute Risk Aversion (CARA) and power Constant Relative Risk Aversion (CRRA) utilities. Reis and Platonov obtain the optimal portfolio using the forward approach, which uses the forward performance process and utilize the supermartingality and martingality conditions of it. They consider the exponential CARA risk preference in \cite{RP1} and the power CRRA risk preference in \cite{RP3}. In this study, we extend the results of Reis and Platonov to the market model with random coefficients which consist of three market parameters $\mu_t,\nu_t,\sigma_t$, and two preference parameters $\delta_t,\theta_t$. To the best of our knowledge, optimal portfolio selection under relative performance concerns in the market model with random coefficients is considered only in \cite{AG1}. They solve the 2-player game, but we solve the $n$-player game and the corresponding mean field stochastic optimal control.

In this article, we define the average wealth of all agents as a new random variable and introduce a forward performance process that depends on its law. Even with the existing utility, $n$-agent game in the market model with random coefficients can be analyzed. See \cite{AG1}. However, it is difficult to extend the $n$-agent game to the continuum of agents by introducing the type vector and type distribution framework used in \cite{LS1,LZ1,RP1,RP3} since the coefficients are random processes. Therefore, we apply the mean field stochastic optimal control theory to portfolio theory. In other words, we take the limit as $n\to\infty$ first, instead of optimizing, to investigate the asymptotic regime, and then optimize over controlled dynamics of McKean-Vlasov type. Here, we apply the theory of the conditional propagation of chaos. In general, finding the Nash equilibrium in the $n$-agent game is challenging. However, the mean field type control problem is tractable to handle since it reduces the dimensionality of the problem. Furthermore, it allows us to obtain an approximate solution of the $n$-agent game. Therefore, we expect to contribute to solving the $n$-agent game, in which it is difficult to find the Nash equilibrium, by presenting a new methodology that extends to the mean field type control problem.

The paper is organized as follows: In Sect. \ref{sec2}, we construct the underlying spaces and recall the auxiliary results regarding the Lions derivative, empirical projection map and the Itô-Wentzell-Lions formula. Then, in Sect. \ref{sec3}, we formulate and solve the $n$-agent model and the corresponding mean field type control problem summarized above when agents have CARA risk preferences. Finally, in Sect. \ref{sec4}, we present the analogous results when agents have CRRA risk preferences.

\section{Underlying spaces and preliminaries}\label{sec2}

Let $(\Omega^0,\mathcal{F}^0,\mathbb{P}^0),\;(\Omega^1,\mathcal{F}^1,\mathbb{P}^1)$ be the atomless separable completely metrizable probability spaces, endowed with two right-continuous and complete filtrations $(\mathcal{F}^0)_{t\geq0},\;(\mathcal{F}^1)_{t\geq0}$. The space $(\Omega^0,\mathcal{F}^0,\mathbb{P}^0)$ is said to be atomless if for any $A\in\mathcal{F}^0$ with $\mathbb{P}^0(A)>0$, there exists $B\in\mathcal{F}^0$, $B\subset A$, such that $0<\mathbb{P}^0(B)<\mathbb{P}^0(A)$. We define the one-dimensional standard Brownian motion $W^0=(W_t^0)_{t\geq0}$ on $(\Omega^0,\mathcal{F}^0,\mathbb{P}^0)$, and the $n$-dimensional standard Brownian motion $W=(W^1,\ldots,W^n)$ on $(\Omega^1,\mathcal{F}^1,\mathbb{P}^1)$, where $W^0$ and $W^i$'s are independent. Let $(\Omega,\mathcal{F},\mathbb{P})$ be the completion of the product space $(\Omega^0\times\Omega^1,\mathcal{F}^0\otimes\mathcal{F}^1,\mathbb{P}^0\otimes\mathbb{P}^1)$. Generic element of $\Omega$ is denoted $\omega=(\omega^0,\omega^1)$ with $\omega^0\in\Omega^0,\;\omega^1\in\Omega^1$.\\
In this study, note that the forward performance process is a random field, not a deterministic function. That is, the dynamics of the forward performance process itself exists, and in order to compute the dynamics of the measure dependent forward relative performance process, the derivative with respect to the measure argument must be defined. There are several notions of differentiability for functions, and we use the one introduced by Lions in \cite{LL1}.\\
Let $\mathcal{P}_2(\mathbb{R})$ be the collection of probability measures defined on $\mathbb{R}$ that have finite second moment, and $u$ be a real-valued function defined on $\mathcal{P}_2(\mathbb{R})$. We consider the \textit{lifting} function $\tilde{u}:L^2(\Omega,\mathcal{F},\mathbb{P};\mathbb{R})\to\mathbb{R}$ of $u$ by
\begin{equation*}
	\tilde{u}(X)=u\big(\mathcal{L}(X)\big),
\end{equation*} where $L^2(\Omega,\mathcal{F},\mathbb{P};\mathbb{R})$ is a space of square-integrable real-valued random variables, and we denote the law of a random variable $Z$ as $\mathcal{L}(Z)$.

\begin{definition}(\textit{Lions}' derivative or L-derivative)\\
	 For a given function $u$ defined on $\mathcal{P}_2(\mathbb{R})$, it is said to be L-differentiable at $\mu\in\mathcal{P}_2(\mathbb{R})$, if the lifting function $\tilde{u}$ is Fr\'echet differentiable at $X$ satisfying $\mu=\mathcal{L}(X)$.
\end{definition}
\begin{remark}
	On the atomless probability space $(\Omega,\mathcal{F},\mathbb{P})$, it is known that for any $\mu\in\mathcal{P}_2(\mathbb{R})$, we can construct a random variable $X:\Omega\to\mathbb{R}$ such that $\mu=\mathcal{L}(X)$. Thus, the lifting of function is well defined on the atomless probability space.  Since $L^2(\Omega,\mathcal{F},\mathbb{P};\mathbb{R})$ is a Hilbert space, we can identify whether the lifting function is Fr\'echet differentiable.
\end{remark}
We denote the corresponding bounded linear operator by $D\tilde{u}(\cdot)$, and $\frac{\partial{u}}{\partial\mu}(\mu)(\cdot)$ as the L-derivative  of $u$, satisfying $D\tilde{u}(X)=\frac{\partial{u}}{\partial\mu}(\mu)(X)$. Here, if the gradient $D\tilde{u}(X)$ is a continuous function of $X$ from the space $L^2(\Omega,\mathcal{F},\mathbb{P};\mathbb{R})$ into itself, then we say that $u$ is continuously L-differentiable.\\
Next, we refer the concept of \textit{empirical projection} of $u$ defined on $\mathcal{P}_2(\mathbb{R})$ given in \cite{CCD1}, and refer to the relationship between its spatial derivative and the L-derivative of $u$ given in \cite[Proposition 5.35 and Proposition 5.91]{CD1}. For $z\in\mathbb{R}$, we denote $\delta_z$ as a Dirac point mass at $z$.
\begin{definition}(Empirical projection of a map)\\
	Given a real-valued function $u$ defined on $\mathcal{P}_2(\mathbb{R})$ and an integer $N\geq1$, then we define the empirical projection $u^{(N)}:\mathbb{R}^N\to\mathbb{R}$ of $u$ by:
	\begin{equation*}
		u^{(N)}(x^1,\ldots,x^N)=u\Big(\frac{1}{N}\sum_{k=1}^N \delta_{x^k}\Big).
	\end{equation*}
\end{definition}

\begin{proposition}\label{prop1}
	Let $u:\mathcal{P}_2(\mathbb{R})\to\mathbb{R}$ be a continuously L-differentiable function, then, for all $N>1$, the empirical projection $u^{(N)}$ is of $\mathcal{C}^2\big(\mathbb{R}^N\big)$. Furthermore, for all $(x^1,\ldots,x^N)\in\mathbb{R}^N,\;i,j\in\{1,\ldots,N\}$, we have the following relationship between spatial derivative of $u^{(N)}$ and the L-derivative of $u$:
	\begin{equation*}
		\frac{\partial u^{(N)}}{\partial{x^i}}(x^1,\ldots,x^N)=\frac{1}{N}\frac{\partial u}{\partial{\mu}}\Big(\frac{1}{N}\sum_{k=1}^{N}\delta_{x^k}\Big)(x^i),
	\end{equation*}
	\begin{align*}
		\frac{\partial^2 u^{(N)}}{\partial x^j\partial{x^i}}(x^1,\ldots,x^N)=\frac{1}{N}&\frac{\partial^2u}{\partial v\partial{\mu}}\Big(\frac{1}{N}\sum_{k=1}^{N}\delta_{x^k}\Big)(x^i)\mathbbm{1}_{\{i=j\}}\\
		&+\frac{1}{N^2} \frac{\partial^2 u}{\partial\mu^2}\Big(\frac{1}{N}\sum_{k=1}^N\delta_{x^k}\Big)(x^i,x^j),
	\end{align*}
	here $\mathbbm{1}_{\{i=j\}}$ is an indicator function that has a value of 1 when superscripts $i=j$ and 0 otherwise.
\end{proposition}
Note that $\frac{\partial^2u}{\partial v\partial{\mu}}(\mu)(v)$ can be considered as a real-valued function defined on $\mathbb{R}$. Thus, we denote the classical spatial derivative of $\frac{\partial u}{\partial{\mu}}(\mu)(v)$ by $\frac{\partial^2u}{\partial v\partial{\mu}}(\mu)(v)$. We refer the reader to \cite[Chapter 5]{CD1} for a detailed explanation of the differentiation of functions of measures.\\
Previous studies \cite{RP1,RP3} used the Itô-Wentzell formula when calculating the dynamics of a forward performance process, but since we deals with a measure dependent forward relative performance process, we use the Itô-Wentzell-Lions formula described below. We refer the reader to \cite{RP2} for a detailed explanation of the Itô-Wentzell-Lions formula.

\begin{proposition}\label{prop2}(Itô-Wentzell-Lions' formula)\\
	Suppose that dynamics of two stochastic processes $X_t,\;Y_t$ are given by
	\begin{equation*}
		\mathrm{d}X_t=b_t\mathrm{d}t+a_t\mathrm{d}W_t+a_t^0\mathrm{d}W_t^0,
	\end{equation*}
	\begin{equation*}
		\mathrm{d}Y_t=\beta_t\mathrm{d}t+\gamma_t\mathrm{d}W_t+\gamma_t^0\mathrm{d}W_t^0,
	\end{equation*}
	where $b_t,\;a_t\;a_t^0\;,\beta_t\;\gamma_t\;\gamma_t^0$ are real-valued $\mathcal{F}$-progressively measurable processes and satisfy
	\begin{equation*}
		\mathbb{E}\Big[\int_0^T\vert b_s\vert+\vert a_s\vert^2+\vert a_s^0\vert^2\Big]<\infty,
	\end{equation*}
	\begin{equation*}
		\mathbb{E}\Big[\int_0^T\vert\beta_s\vert^2+\vert\gamma_s\vert^4+\vert\gamma_s^0\vert^4\Big]<\infty.
	\end{equation*}
	Furthermore, dynamics of the given random field $u:\Omega\times(0,\infty)\times\mathcal{P}_2(\mathbb{R})\times[0,\infty)\to\mathbb{R}$ is given by
	\begin{equation*}
		\mathrm{d}u(x,\mu,t)=\phi(x,\mu,t)dt+\psi(x,\mu,t)\mathrm{d}W_t+\psi^0(x,\mu,t)\mathrm{d}W_t^0.
	\end{equation*}
	Let $\mu_t(\omega^0)=\mathcal{L}\big(Y_t(\omega^0,\cdot)\big)$, then under sufficient regularity and local boundedness, $\mathbb{P}^0$-almost surely, dynamics of $u(X_t,\mu_t,t)$ is given by
	\begin{align*}
		\begin{split}
			\mathrm{d}u(X_t,\mu_t,t)=\;&\phi(X_t,\mu_t,t)\mathrm{d}t+\psi(X_t,\mu_t,t)\mathrm{d}W_t+\psi^0(X_t,\mu_t,t)\mathrm{d}W_t^0\\
			&+\frac{\partial u}{\partial x}(X_t,\mu_t,t)\Big(b_t\mathrm{d}t+a_t\mathrm{d}W_t+a_t^0\mathrm{d}W_t^0\Big)\\
			&+\frac{1}{2}\frac{\partial^2 u}{\partial x^2}(X_t,\mu_t,t)\Big(\vert a_t\vert^2+\vert a_t^0\vert^2\Big)\mathrm{d}t\\
			&+\widetilde{\mathbb{E}}^1\bigg[\frac{\partial u}{\partial \mu}(X_t,\mu_t,t)(\widetilde{Y_t})\widetilde{\beta_t}\bigg]\mathrm{d}t\\
			&+\widetilde{\mathbb{E}}^1\bigg[\frac{\partial u}{\partial \mu}(X_t,\mu_t,t)(\widetilde{Y_t})\widetilde{\gamma_t^{0}}\bigg]\mathrm{d}W_t^0\\
			&+\frac{1}{2}\widetilde{\mathbb{E}}^1\bigg[\frac{\partial^2 u}{\partial v \partial \mu}(X_t,\mu_t,t)(\widetilde{Y_t})\big(\vert\widetilde{\gamma_t}\vert^2+\vert\widetilde{\gamma_t^{0}}\vert^2\big)\bigg]\mathrm{d}t\\
			&+\frac{1}{2}\widehat{\mathbb{E}}^1\bigg[\widetilde{\mathbb{E}}^1\bigg[\frac{\partial^2 u}{\partial \mu^2}(X_t,\mu_t,t)(\widetilde{Y_t},\widehat{Y_t})\widetilde{\gamma_t^{0}}\widehat{\gamma_t^{0}}\bigg]\bigg]\mathrm{d}t\\
			&+\widetilde{\mathbb{E}}^1\bigg[\frac{\partial^2 u}{\partial x \partial \mu}(X_t,\mu_t,t)(\widetilde{Y_t})a_t^0\widetilde{\gamma_t^{0}}\bigg]dt+\frac{\partial \psi^0}{\partial x}(X_t,\mu_t,t)a_t^0\mathrm{d}t\\
			&+\frac{\partial \psi}{\partial x}(X_t,\mu_t,t)a_t\mathrm{d}t+\widetilde{\mathbb{E}}^1\bigg[\frac{\partial \psi^0}{\partial \mu}(X_t,\mu_t,t)(\tilde{Y_t})\widetilde{\gamma_t^0}\bigg]\mathrm{d}t.
		\end{split}
	\end{align*}
	Note that $\mu_t(\omega^0)$ is a probability measure on $\Omega^1$. More precisely, it is the law of $Y_t$ when the common noise is realized.
\end{proposition}
  $(\widetilde{\Omega}^1,\widetilde{\mathcal{F}}^1,\widetilde{\mathbb{P}}^1)$, $(\widehat{\Omega}^1,\widehat{\mathcal{F}}^1,\widehat{\mathbb{P}}^1)$ are copy spaces of $(\Omega^1,\mathcal{F}^1,\mathbb{P}^1)$, and $(\widetilde{\Omega},\widetilde{\mathcal{F}},\widetilde{\mathbb{P}})$, $(\widehat{\Omega},\widehat{\mathcal{F}},\widehat{\mathbb{P}})$ are completion of the product space $(\Omega^0\times\widetilde{\Omega}^1,\mathcal{F}^0\otimes\widetilde{\mathcal{F}}^1,\mathbb{P}^0\otimes\widetilde{\mathbb{P}}^1)$, $(\Omega^0\times\widehat{\Omega}^1,\mathcal{F}^0\otimes\widehat{\mathcal{F}}^1,\mathbb{P}^0\otimes\widehat{\mathbb{P}}^1)$, respectively, and $(\widetilde{Y_t},\widetilde{\beta_t},\widetilde{\gamma_t},\widetilde{\gamma_t^{0}})$, $(\widehat{Y_t},\widehat{\beta_t},\widehat{\gamma_t},\widehat{\gamma_t^{0}})$ are independent copy processes of $({Y_t},{\beta_t},{\gamma_t},{\gamma_t^{0}})$ defined on copy spaces $(\widetilde{\Omega},\widetilde{\mathcal{F}},\widetilde{\mathbb{P}})$, $(\widehat{\Omega},\widehat{\mathcal{F}},\widehat{\mathbb{P}})$, respectively. $\widetilde{\mathbb{E}},\;\widehat{\mathbb{E}}$ denote the expectation acting on copy spaces $(\widetilde{\Omega},\widetilde{\mathcal{F}},\widetilde{\mathbb{P}})$, $(\widehat{\Omega},\widehat{\mathcal{F}},\widehat{\mathbb{P}})$, respectively. Note that $\widetilde{\mathbb{E}}^1$, $\widehat{\mathbb{E}}^1$ denote the conditional expectation given $\mathcal{F}^0$ depending on the probability space.

\section{CARA risk preferences}\label{sec3}
We investigate agents who have exponential risk preferences with random individual absolute risk tolerances and absolute competition weights. Each agent measures its relative performance, taking into account competition among agents, and using the difference between their own wealth and the average wealth of all agents as a benchmark.
\subsection{The $n$-agent game}\label{subsec3.1}
We introduce a game of $n$ agents competing with each other. The market consists of one riskless asset and $n$ risky assets. We assume asset specialization in which each agent $i$ can invest in her assigned risky asset $S^i$ and a common riskfree asset. We also assume that price of risky assets follows a log-normal distribution, and each risky asset $S^i$ is driven by two independent Brownian motions $W^i$ and $W^0$. Precisely, dynamics of the risky asset price $(S_t^i)_{t\geq0}$ traded exclusively by the $i$-th agent is given by
\begin{equation*}
	\frac{\mathrm{d}S_t^i}{S_t^i}=\breve{\mu}(S_t^i,t,\omega)\mathrm{d}t+\nu(S_t^i,t,\omega)\mathrm{d}W_t^i+\sigma(S_t^i,t,\omega)\mathrm{d}W_t^0,
\end{equation*}
where $\omega\in\Omega$, $\breve{\mu},\nu,\sigma$ is $\mathcal{F}$-progressively measurable processes with $\breve{\mu}>0,\;\sigma\geq0,\;\nu\geq0,\;\Sigma=\sigma^2+\nu^2>0$, $\mathbb{P}$-almost surely. When $\sigma>0$, the Brownian motion $W^0$ induces a correlation among risky assets, so we call $W^0$ the \textit{common noise} and $W^i$ an \textit{idiosyncratic noise}. Let $\mathcal{F}$-progressively measurable process $r(t,\omega)$ be the interest rate and denote $\mu=\breve{\mu}-r$. Then $\mu$ is an $\mathcal{F}$-progressively measurable process, and it can be interpreted as an excess return. For convenience, we omit the omega argument $\omega$ for the following description, and specify if it is to be considered.\\
We suppose that each agent $i\in\{1,\ldots,n\}$ trades using a self-financing strategy $(\pi_t^i)_{t\geq0}$, which represents the (discounted by the bond) amount invested in the $i$-th risky asset. Then dynamics of the $i$-th agent's wealth process $(X_t^i)_{t\geq0}$ is given by
\begin{equation}\label{eq:1}
	\mathrm{d}X_t^i=\pi_t^i\big(\mu_t^i\mathrm{d}t+\nu_t^i\mathrm{d}W_t^i+\sigma_t^i\mathrm{d}W_t^0\big),\;\;X_0^i=x_0^i\in\mathbb{R},
\end{equation}
where
\begin{equation*}
	\mu_t^i=\mu(X_t^i,t),\;\nu_t^i=\nu(X_t^i,t),\;\sigma_t^i=\sigma(X_t^i,t),\;\pi_t^i=\pi\big(X_t^i,\frac{1}{n}\sum_{k=1}^n \delta_{X_t^k},t\big).
\end{equation*}
Next, we define the admissibility set $\mathcal{A}^i$ for the agent $i$ to be the collection of $\mathcal{F}$-progressively measurable processes $(\pi_t^i)_{t\geq0}$, such that
\begin{equation*}
	\mathbb{E}\big[\int_0^t\vert\pi_s^i\mu_s^i\vert^2ds\big]<\infty,\;\mathbb{E}\big[\int_0^t\vert\pi_s^i\nu_s^i\vert^4ds\big]<\infty,\;\mathbb{E}\big[\int_0^t\vert\pi_s^i\sigma_s^i\vert^4ds\big]<\infty,
\end{equation*}
 for any $t>0$. We say that a portfolio strategy $\pi_t^i$ for the agent $i$ is \textit{admissible} if it belongs to $\mathcal{A}^i$.\\

Considering competition among agents in order to measure a relative performance, we introduce a stochastic process $(Y_t^n)_{t\geq0}$ by
\begin{equation*}
	Y_t^n=h(\overline{M}_t^n),
\end{equation*} where
\begin{equation*}
	\overline{M}_t^n=\frac{1}{n}\sum_{k=1}^n \delta_{X_t^k},
\end{equation*}
and
\begin{equation*}
	h(\mu)=\int_\mathbb{R} x\;\mathrm{d}\mu(x).
\end{equation*}
 Then $Y_t^n=\frac{1}{n}\sum_{k=1}^nX_t^k$, and it is the arithmetic average wealth of all agents. We can calculate the dynamics of $Y_t^n$ in two ways. Firstly, we can derive it from the general method $\mathrm{d}Y_t^n=\frac{1}{n}\sum_{k=1}^n\mathrm{d}X_t^k$. Secondly, we can use the empirical projection $h^{(n)}(X_t^1,\ldots,X_t^n)$ and derive it from $\mathrm{d}Y_t^n=\mathrm{d}h^{(n)}(X_t^1,\ldots,X_t^n)$. After some calculations using Proposition \ref{prop1}, the dynamics of $Y_t^i$ is given by
\begin{equation*}
	\mathrm{d}Y_t^n=(\overline{\pi\mu})_t\mathrm{d}t+\frac{1}{n}\sum_{k=1}^n\pi_t^k\nu_t^k\mathrm{d}W_t^k+(\overline{\pi\sigma})_t\mathrm{d}W_t^0,
\end{equation*}
where we define the auxiliary quantities 
\begin{equation*}
	(\overline{\pi\mu})_t=\frac{1}{n}\sum_{k=1}^n\pi_t^k\mu_t^k,
\end{equation*}
\begin{equation*}
	(\overline{\pi\sigma})_t=\frac{1}{n}\sum_{k=1}^n\pi_t^k\sigma_t^k.
\end{equation*}
We define a stochastic process $\alpha^n:\Omega^0\times[0,\infty)\to\mathcal{P}_2(\mathbb{R})$ by
\begin{equation*}
	\alpha_t^n(\omega^0)=\mathcal{L}\big(Y_t^n(\omega^0,\cdot)\big).
\end{equation*} 

\begin{remark}
	Similary to \cite[Remark 2.5]{LZ1}, It is more natural to replace the average wealth $Y_t^n$ with the average over all other agents. If we define $Y_t^i$ by
	\begin{equation*}
		Y_t^i=p_n(\overline{M}_t^n)+q_n(X_t^i)
	\end{equation*}
	where
	\begin{equation*}
		p_n(\mu)=\int_\mathbb{R} \frac{n}{n-1}x\;\mathrm{d}\mu(x)
	\end{equation*}
	and
	\begin{equation*}
		q_n(x)=-\frac{1}{n-1}x,
	\end{equation*}
	then we can obtain the corresponding results for that case.
\end{remark}

 We assume that the $i$-th agent's utility is a random field $U^i:\Omega\times(0,\infty)\times\mathcal{P}_2(\mathbb{R})\times[0,\infty)\to\mathbb{R}$ such that
 \begin{equation}\label{eq:2}
 	U^i(x,\mu,t)=-\exp\Big[-\frac{1}{\delta_t^i}\big(x-\theta_t^i \lambda(\mu)\big)+K_t^i\Big],
 \end{equation}
 where $K_t^i$ is an $\mathcal{F}$-progressively measurable process that is differentiable in time with $K_0^i=0$, and we define
 \begin{equation*}
 	\lambda(\mu)=\int_\mathbb{R}x\;\mathrm{d}\mu(x).
 \end{equation*}
Here, the random parameters satisfy the conditions $\delta_t^i=\delta(X_t^i,t)>0$, $\theta_t^i=\theta(X_t^i,t)\in[0,1]$, $\mathbb{P}$-almost surely, and they represent the $i$-th agent's absolute risk tolerance and absolute competition weight. Note that
\begin{equation*}
	U^i(X_t^i,\alpha_t^n,t)=-\exp\Big[-\frac{1}{\delta_t^i}\big(X_t^i-\theta_t^i\overline{X}_t\big)+K_t^i\Big]
\end{equation*}
where
\begin{equation*}
	\overline{X}_t=\frac{1}{n}\sum_{k=1}^n X_t^k,
\end{equation*}
which is almost the same as the CARA exponential utility form used in \cite{LZ1,RP1}.\\
Note that
\begin{equation*}
	X_t^i-\theta_t^i\overline{X}_t=(1-\theta_t^i)X_t^i+\theta_t^i(X_t^i-\overline{X}_t).
\end{equation*}
Thus, the smaller value of $\theta_t^i$, the relative performance becomes less relevant.\\

For each agent $i$, we also assume that the Itô-decomposition of the utility is given by
\begin{equation}\label{eq:3}
	\mathrm{d}U^i(x,\mu,t)=\frac{\partial U^i}{\partial t}(x,\mu,t)\mathrm{d}t.
\end{equation}
We may assume that the volatility of the utility is zero since we only work in a lognormal market.\\

We can check that partial derivatives and L-derivatives of $U^i$ exist and they are given in the following Proposition.
\begin{proposition}\label{prop3}
	For each agent $i\in\{1,\ldots,n\}$, the utility $U^i$ has the partial derivatives and L-derivatives given as follows:
	\begin{align*}
		\frac{\partial U^i}{\partial t}(x,\mu,t)&=\frac{\mathrm{d}K_t^i}{\mathrm{d}t}U^i(x,\mu,t),\\
		\frac{\partial U^i}{\partial x}(x,\mu,t)&=-\frac{1}{\delta_t^i}U^i(x,\mu,t),\\
		\frac{\partial^2U^i}{\partial x^2}(x,\mu,t)&=\frac{1}{(\delta_t^i)^2}U^i(x,\mu,t),\\
		\frac{\partial U^i}{\partial \mu}(x,\mu,t)(v)&=\frac{\theta_t^i}{\delta_t^i}U^i(x,\mu,t),\\
		\frac{\partial^2U^i}{\partial v \partial \mu}(x,\mu,t)(v)&=0,\\
		\frac{\partial^2U^i}{\partial \mu^2}(x,\mu,t)(v,v')&=\Big(\frac{\theta_t^i}{\delta_t^i}\Big)^2U^i(x,\mu,t),\\
		\frac{\partial^2U^i}{\partial x \partial \mu}(x,\mu,t)(v)&=-\frac{\theta_t^i}{(\delta_t^i)^2}U^i(x,\mu,t).
	\end{align*}
\end{proposition}
\begin{proof}[Proof]
	It is known that the L-derivative of $u(\mu)$ of the form $\int_{\mathbb{R}}f(x)\mathrm{d}\mu(x)$ is given by $\frac{\partial f}{\partial x}$ if $\frac{\partial f}{\partial x}(x)$ is at most of linear growth. See \cite[Section 5.2.2]{CD1}.\\
	Since $\frac{\partial \lambda}{\partial \mu}(\mu)(v)=1$, it implies that $\frac{\partial U^i}{\partial \mu}(x,\mu,t)(v)=\frac{\theta_t^i}{\delta_t^i}U^i(x,\mu,t)$. Note that, it does not depend on the new variable $v$, and we automatically obtain $\frac{\partial^2U^i}{\partial v \partial \mu}(x,\mu,t)$=0. Other derivatives can be obtained using the usual differentiation method.
\end{proof}

We now refer the concept of CARA-\textit{forward relative performance} in \cite{RP1}, and extend it to the measure dependent forward relative performance process.
\begin{definition}(CARA-forward relative performance for the agent)\label{def3}\\
	We assume that all agents have CARA risk preferences. For each agent $i$, we suppose that a strategy for the other agent $j$, $\pi_t^j$, is arbitrary but fixed, for all $j\neq i$. The $\mathcal{F}$-progressively measurable random field $U^i:\Omega\times(0,\infty)\times\mathcal{P}_2(\mathbb{R})\times[0,\infty)\to\mathbb{R}$ is a forward relative performance for the agent $i$ if, for all $t\geq0$, the followings hold:
	\begin{enumerate}
		\item $U^i(x,\mu,t)$ is $\mathbb{P}$-a.s. strictly increasing and strictly concave in $x$.
		\item For any $\pi_t^i\in\mathcal{A}^i$, $U^i(X_t^i,\alpha_t^n,t)$ is a $\mathbb{P}^1$-(local) supermartingale $\mathbb{P}^0$-a.s.
		\item There exists $\pi_t^{i,*}\in\mathcal{A}^i$, such that $U^i(X_t^{i,*},\alpha_t^n,t)$ is a $\mathbb{P}^1$-(local) martingale $\mathbb{P}^0$-a.s. where $X_t^{i,*}$ satisfies (\ref{eq:1}) with a strategy $\pi_t^{i,*}$.
	\end{enumerate}
\end{definition}

The above definition assumes the optimal strategy is attained. Actually, it strongly depends on the initial condition $U^i(x,\mu,0)$. Thus, the analysis of an admissible initial condition is also valuable as a future research.\\
Contrary to the backward approach based on the classical expected utility maximization case, the forward relative performance process is a manager-specific input. Once it is chosen, the supermatingale and martingale conditions induce the forward stochastic partial differential equation, with enough regularity. See \cite{MZ3,NM1}. In our case, it reduces to an ODE with stochastic coefficients.

\begin{remark}
	We can also extend the concept of {forward performance process} in \cite{MZ2} to the measure dependent forward performance process. It is defined by substituting $\delta_{0}$ instead of $\alpha_t^n$ in the argument of the utility in Definition \ref{def3}. Here, we denote $\delta_0$ as a Dirac point mass at $0\in\mathbb{R}$.
\end{remark}

We compare the previous study with our study in the below table \ref{tab1}. Specifically, we compare the form of the forward performance process and the forward relative performance used in each study and stochastic process which is analyzed to consider competition.\\

\begin{table}[h]
	\begin{center}
		\begin{tabular} {ccccc}
			\hline\hline\addlinespace[0.1cm]
			& &\textbf{Previous study} & \textbf{Our study} &\\[0.1cm]
			\hline\addlinespace[0.2cm]
			& \textbf{Forward performance process}      & $U^i(X_t^i,t)$ & $U^i(X_t^i,\delta_0,t)$ & \\[0.2cm]
			\hline\addlinespace[0.2cm]
			& \textbf{Forward relative performance process} & $U^i(X_t^i-\theta_t^i\overline{X}_t,t)$ &     $U^i(X_t^i,\alpha_t^n,t)$ &  \\[0.2cm]
			\hline\hline
		\end{tabular}
	\end{center}
	\caption[Comparison of the forward relative performance process formula]{Comparison of the forward relative performance process formula}
	\label{tab1}
\end{table}

When the other agents' strategies are arbitrary but fixed, we can reply to this strategy to obtain the agent's own optimal investment policy. Here we derive the random ODE which the $K_t^i$ satisfy in order for $U^i$ to be a forward relative performance, and investigate the optimal strategy for each agent.\\

\begin{proposition}\label{prop4} (Best responses)\\
	Fix $i\in\{1,\ldots,n\}$ and suppose that $\pi_t^j$ for the other agents $j\neq i$ are arbitrary but fixed. We assume that $U^i$, the utility of agent $i$, satisfies conditions (\ref{eq:2}), (\ref{eq:3}) for given $K_t^i$. Consider the random ODE,
	\begin{equation}\label{eq:4}
		\frac{\mathrm{d}K_t^i}{\mathrm{d}t}=-\frac{\theta_t^i}{\delta_t^i}\mathbb{E}^1[(\overline{\pi\mu})_t]-\frac{1}{2}\Big(\frac{\theta_t^i}{\delta_t^i}\Big)^2\frac{(\nu_t^i)^2}{\Sigma_t^i}\mathbb{E}^1[(\overline{\pi\sigma})_t]^2+\frac{\theta_t^i}{2\delta_t^i}\frac{\mu_t^i\sigma_t^i}{\Sigma_t^i}\mathbb{E}^1[(\overline{\pi\sigma})_t]+\frac{(\mu_t^i)^2}{2\Sigma_t^i},
	\end{equation}
	and define the strategy $\pi_t^{i,*}$ by
	\begin{equation}\label{eq:5}
		\pi_t^{i,*}=\frac{1}{\Sigma_t^i}\Big(\mu_t^i\delta_t^i+\theta_t^i\sigma_t^i\mathbb{E}^1\big[(\overline{\pi\sigma})_t\big]\Big),\;\; t\geq0.
	\end{equation}
	If $K_t^i$ satisfies (\ref{eq:4}) $\mathbb{P}^0$-a.s., then $U^i(x,\mu,t)$ is a forward relative performance process, and the policy $\pi_t^{i,*}$ is an optimal strategy for an agent $i$.
\end{proposition}

\begin{proof}[Proof]
	Since we deals with a measure dependent forward relative performance process, we use the Itô-Wentzell-Lions formula in Proposition \ref{prop2}. Then we have $\mathbb{P}^0$-a.s.,
	\begin{align}\label{eq:6}
		\begin{split}
			\mathrm{d}U^i(X_t^i,\alpha_t^n,t)=\;&\frac{\partial U^i}{\partial t}(X_t^i,\alpha_t^n,t)\mathrm{d}t+\frac{\partial U^i}{\partial x}(X_t^i,\alpha_t^n,t)\pi_t^i\big(\mu_t^i\mathrm{d}t+\nu_t^i\mathrm{d}W_t^i+\sigma_t^i\mathrm{d}W_t^0\big)\\
			&+\frac{1}{2}\frac{\partial^2U^i}{\partial x^2}(X_t^i,\alpha_t^n,t)(\pi_t^i)^2\Sigma_t^i\mathrm{d}t\\
			&+\widetilde{\mathbb{E}}^1\bigg[\frac{\partial U^i}{\partial \mu}(X_t^i,\alpha_t^n,t)(\widetilde{Y}_t^n)\overline{(\widetilde{\pi\mu})}_t\bigg]\mathrm{d}t\\
			&+\widetilde{\mathbb{E}}^1\bigg[\frac{\partial U^i}{\partial \mu}(X_t^i,\alpha_t^n,t)(\widetilde{Y}_t^n)\overline{(\widetilde{\pi\sigma})}_t\bigg]\mathrm{d}W_t^0\\
			&+\frac{1}{2}\widetilde{\mathbb{E}}^1\bigg[\frac{\partial^2U^i}{\partial v \partial \mu}(X_t^i,\alpha_t^n,t)(\widetilde{Y}_t^n)\Big(\frac{1}{n^2}\sum_{k=1}^n\big(\tilde{\pi}_t^k\tilde{\nu}_t^k\big)^2+\big(\overline{(\widetilde{\pi\sigma})}_t\big)^2\Big)\bigg]\mathrm{d}t\\
			&+\frac{1}{2}\widehat{\mathbb{E}}^1\bigg[\widetilde{\mathbb{E}}^1\bigg[\frac{\partial^2U^i}{\partial \mu^2}(X_t^i,\alpha_t^n,t)(\widetilde{Y}_t^n,\widehat{Y}_t^n)\overline{(\widetilde{\pi\sigma})}_t\overline{(\widehat{\pi\sigma})}_t\bigg]\bigg]\mathrm{d}t\\
			&+\widetilde{\mathbb{E}}^1\bigg[\frac{\partial^2U^i}{\partial x \partial \mu}(X_t^i,\alpha_t^n,t)(\widetilde{Y}_t^n)\pi_t^i\sigma_t^i\overline{(\widetilde{\pi\sigma})}_t\bigg]\mathrm{d}t.
		\end{split}
	\end{align}
	Substituting the derivatives in Proposition \ref{prop3} into (\ref{eq:6}), we obtain $\mathbb{P}^0$-a.s.,
	
\begin{align}\label{eq:7}
	\begin{split}
		\frac{\mathrm{d}U^i(X_t^i,\alpha_t^n,t)}{U^i(X_t^i,\alpha_t^n,t)}=\;&\bigg[\frac{\Sigma_t^i}{2(\delta_t^i)^2}(\pi_t^i)^2-\frac{1}{(\delta_t^i)^2}\Big(\mu_t^i\delta_t^i+\theta_t^i\sigma_t^i\mathbb{E}^1[(\overline{\pi\sigma})_t]\Big)\pi_t^i\\
		&\qquad\qquad\qquad+\frac{\theta_t^i}{\delta_t^i}\Big(\mathbb{E}^1[(\overline{\pi\mu})_t]+\frac{\theta_t^i}{2\delta_t^i}\mathbb{E}^1[(\overline{\pi\sigma})_t]^2\Big)+\frac{\mathrm{d}K_t^i}{\mathrm{d}t}\bigg]\mathrm{d}t\\
		&-\frac{1}{\delta_t^i}\pi_t^i\nu_t^i\mathrm{d}W_t^i-\frac{1}{\delta_t^i}\Big(\pi_t^i\sigma_t^i-\theta_t^i\mathbb{E}^1[(\overline{\pi\sigma})_t]\Big)\mathrm{d}W_t^0.
	\end{split}
&\end{align}
For the process $U^i(x,\mu,t)$ to be a forward relative performance, a value of the drift term in (\ref{eq:7}) should be less than or equal to 0 for all $\pi_t^i\in\mathcal{A}^i$, and 0 for an optimal strategy. Note that the drift term in (\ref{eq:7}) is a quadratic function with respect to $\pi_t^i$ where the quadratic coefficient is negative. If $K_t^i$ satisfies (\ref{eq:4}) and we define $\pi_t^{i,*}$ as in (\ref{eq:5}), then we can simplify the drift term to $\frac{\Sigma_t^i}{2(\delta_t^i)^2}U^i\vert \pi_t^i-\pi_t^{i,*}\vert^2$. Thus, we conclude that $\pi_t^{i,*}$ is an optimal strategy and $U^i(x,\mu,t)$ is a forward relative performance for agent $i$.
\end{proof}

We refer the concept of CARA-\textit{forward Nash equilibrium} in \cite{RP1}, and extend it to the measure dependent forward relative performance process.

\begin{definition}(CARA-forward Nash equilibrium)\label{def4}
	We assume that each agent $i$ has CARA risk preference $U^i:\Omega\times(0,\infty)\times\mathcal{P}_2(\mathbb{R})\times[0,\infty)\to\mathbb{R}$, $1\leq i\leq n$. A forward Nash equilibrium consists of $n$-pairs of $\mathcal{F}$-progressively measurable random fields $\big(U^i,\pi^{i,*}\big)$, $1\leq i\leq n$, such that, for any $t\geq0$, the following conditions hold:
	\begin{enumerate}
		\item For each $i$, $\pi_t^{i,*}\in\mathcal{A}^i$.
		\item For each $i$, $U^i(x,\mu,t)$ is $\mathbb{P}$-a.s. strictly increasing and strictly concave in $x$.
		\item Let agents $j\neq i$ invest along the strategy $\pi_t^{j,*}$ and agent $i$ invests along the strategy $\pi_t^i\in\mathcal{A}^i$. Then $U^i(X_t^i,\alpha_t^n,t)$ is a $\mathbb{P}^1$-(local) supermartingale $\mathbb{P}^0$-a.s.
		\item Let all agents $i$ invest along the strategy $\pi_t^{i,*}$ and $X_t^{i,*}$ be the associated wealth process and $\alpha_t^{n,*}$ be the associated conditional marginal distribution. Then $U^i(X_t^{i,*},\alpha_t^{n,*},t)$ is a $\mathbb{P}^1$-(local) martingale $\mathbb{P}^0$-a.s.
	\end{enumerate}
\end{definition}

We now present the first main result, in which we find the forward Nash equilibrium when agents have CARA risk preferences. We solve the simultaneous best response as to construct the forward Nash equilibrium.

\begin{theorem}\label{thm5}
	Assume the conditions of Proposition \ref{prop4} hold for all agents $i\in\{1,\ldots,n\}$. Define the random quantities
	\begin{equation*}
		\varphi_n^\sigma(t)=\mathbb{E}^1\Big[\frac{1}{n}\sum_{k=1}^n\frac{\mu_t^k\delta_t^k\sigma_t^k}{\Sigma_t^k}\Big],
	\end{equation*}
	\begin{equation*}
		\psi_n^\sigma(t)=\mathbb{E}^1\Big[\frac{1}{n}\sum_{k=1}^n\frac{\theta_t^k(\sigma_t^k)^2}{\Sigma_t^k}\Big].
	\end{equation*}
	If $\psi_n^\sigma(t)\neq 1$ $\mathbb{P}^0$-a.s. for every $t\geq0$, then there exists an optimal strategy $\pi_t^{i,*}$ defined by
	\begin{equation*}
		\pi_t^{i,*}=\frac{1}{\Sigma_t^i}\Big(\mu_t^i\delta_t^i+\theta_t^i\sigma_t^i\frac{\varphi_n^\sigma(t)}{1-\psi_n^\sigma(t)}\Big),\;t\geq0.
	\end{equation*}
	Furthermore, let
	\begin{equation}\label{eq:8}
		K_t^i=-\displaystyle\int_0^t\frac{\theta_s^i}{\delta_s^i}\mathbb{E}^1[(\overline{\pi\mu})_s]+\frac{1}{2}\Big(\frac{\theta_s^i}{\delta_s^i}\Big)^2\frac{(\nu_s^i)^2}{\Sigma_s^i}\mathbb{E}^1[(\overline{\pi\sigma})_s]^2-\frac{\theta_s^i}{2\delta_s^i}\frac{\mu_s^i\sigma_s^i}{\Sigma_s^i}\mathbb{E}^1[(\overline{\pi\sigma})_s]-\frac{(\mu_s^i)^2}{2\Sigma_s^i}\mathrm{d}s
	\end{equation}
	where
	\begin{equation*}
		\mathbb{E}^1[(\overline{\pi\sigma})_t]=\frac{\varphi_n^\sigma(t)}{1-\psi_n^\sigma(t)},
	\end{equation*}
	\begin{equation*}
		\mathbb{E}^1[(\overline{\pi\mu})_t]=\mathbb{E}^1\Big[\frac{1}{n}\sum_{k=1}^n\frac{(\mu_t^k)^2\delta_t^k}{\Sigma_t^k}\Big]+\mathbb{E}^1\Big[\frac{1}{n}\sum_{k=1}^n\frac{\mu_t^k\theta_t^k\sigma_t^k}{\Sigma_t^k}\Big]\frac{\varphi_n^\sigma(t)}{1-\psi_n^\sigma(t)}.
	\end{equation*}
	 Then $n$-pairs of $\big(U^{i,*},\pi^{i,*}\big)$ are a forward Nash equilibrium, where $U^{i,*}$ satisfies (\ref{eq:3}) with $X_t^{i,*},\;\alpha_t^{n,*},\;K_t^i$.
\end{theorem}

\begin{proof}[Proof]
	Multiplying both sides of (\ref{eq:5}) by $\sigma_t^i$, and averaging over $i$'s, then
	\begin{equation*}
		(\overline{\pi\sigma})_t=\frac{1}{n}\sum_{k=1}^n\frac{\mu_t^k\delta_t^k\sigma_t^k}{\Sigma_t^k}+\frac{1}{n}\sum_{k=1}^n\frac{\theta_t^k(\sigma_t^k)^2}{\Sigma_t^k}\mathbb{E}^1[(\overline{\pi\sigma})_t].
	\end{equation*}
	Now, taking conditional expectations on both sides, we obtain
	\begin{equation*}
		\mathbb{E}^1[(\overline{\pi\sigma})_t]=\frac{\varphi_n^\sigma(t)}{1-\psi_n^\sigma(t)}.
	\end{equation*}
	Substituting it into (\ref{eq:5}), we obtain a desired optimal strategy $\pi_t^{i,*}$.\\
	Furthermore, if we let $K_t^i$ be given by (\ref{eq:8}), then it satisfies (\ref{eq:4}).
	Thus, $U^i(x,\mu,t)$ is a forward relative performance process. We can easily check that $\big(U^{i,*},\pi^{i,*}\big)$ satisfies conditions given in Definition \ref{def4}. Thus, $n$-pairs of $\big(U^{i,*},\pi^{i,*}\big)$ are a forward Nash equilibrium.\\
	Similarly, we have
	\begin{equation*}
		\mathbb{E}^1[(\overline{\pi\mu})_t]=\mathbb{E}^1\Big[\frac{1}{n}\sum_{k=1}^n\frac{(\mu_t^k)^2\delta_t^k}{\Sigma_t^k}\Big]+\mathbb{E}^1\Big[\frac{1}{n}\sum_{k=1}^n\frac{\mu_t^k\theta_t^k\sigma_t^k}{\Sigma_t^k}\Big]\frac{\varphi_n^\sigma(t)}{1-\psi_n^\sigma(t)}.
	\end{equation*}
\end{proof}

\begin{remark}
	We analyze the case where the value of $\psi_n^\sigma(t)$ becomes $1$. Note that $\theta_t^k\in[0,1]$, $(\delta_t^k)^2\leq\Sigma_t^k$, for all $k\in\{1,\ldots,n\}$, $\mathbb{P}$-a.s. Thus, $\psi_n^\sigma(t)\leq1$, $\mathbb{P}$-a.s. We conclude that $\psi_n^\sigma(t)=1$ if and only if $\theta_t^k=1,\;\nu_t^k=0$, for all $k\in\{1,\ldots,n\}$, $\mathbb{P}$-a.s. This is the case when all agents invest in a common risky asset, and take into account competition among agents as much as possible.
\end{remark}

\subsection{The mean field stochastic optimal control}\label{subsec3.2}

In this Section, we study the limit as $n\to\infty$ of the $n$-agent game analyzed in previous Section. As the number of players approaches infinity, there are two methods to investigate the asymptotic regime of stochastic differential games with a finite number of players. The difference between the two methods is whichever comes first, optimization or performing passage to the limit.\\
In previous studies \cite{LS1,LZ1,RP1,RP3},  they introduce the concepts of type vector and type distribution and deal with the continuum of agents game by constructing a mean-field game. However, since we consider the random coefficients rather than the constant coefficients, analysis using the type vector and type distribution is not obvious at this point. Therefore, we use a method different from the existing one.\\
We perform passaging to the limit $n\to\infty$ first to investigate the asymptotic regime, and then optimize over controlled dynamics of McKean-Vlasov type.
We define the stochastic processes $(\underline{X}_t^i)_{t\geq0}$, for $i\in\{1,\ldots,n\}$, satisfy the following McKean-Vlasov SDE,
\begin{equation*}
	\mathrm{d}\underline{X}_t^i=\pi\big(\underline{X}_t^i,\mathcal{L}^1(\underline{X}_t^i),t\big)\Big[\mu(\underline{X}_t^i,t)\mathrm{d}t+\nu(\underline{X}_t^i,t)\mathrm{d}W_t^i+\sigma(\underline{X}_t^i,t)\mathrm{d}W_t^0\Big].
\end{equation*}
Then the theory of conditional propagation of chaos states that, for all $i\geq1$, the followings hold:
\begin{equation*}
	\mathbb{P}^0\big\{\mathcal{L}^1(\underline{X}_t^i)=\mathcal{L}^1(\underline{X}_t^1),\;\text{for any}\;t\geq0\big\}=1,
\end{equation*}
\begin{equation*}
	\lim_{n\to\infty}\sup_{t\geq0}\mathbb{E}\big[W_2(\overline{M}_t^n,\mathcal{L}^1(\underline{X}_t^1))^2\big]=0,
\end{equation*}
where $W_2$ is the 2-Wasserstein distance, which is a metric on $\mathcal{P}_2(\mathbb{R}).$ We refer the reader to \cite[Section 2.1]{CD2}, for a detailed explanation of the theory of conditional propagation of chaos.\\
Therefore, we can consider a representative player $X=(X_t)_{t\geq0}$ such that
\begin{equation*}
	\mathbb{P}^0\big\{\mathcal{L}^1(X_t)=\mathcal{L}^1(\underline{X}_t^1),\;\text{for any}\;t\geq0\big\}=1,
\end{equation*}
and its dynamics is given by
\begin{equation}\label{eq:9}
	\mathrm{d}X_t=\pi_t\big(\mu_t\mathrm{d}t+\nu_t\mathrm{d}W_t+\sigma_t\mathrm{d}W_t^0\big),
\end{equation}
 where $W$ is a one-dimensional standard Brownian motion defined on $(\Omega^1,\mathcal{F}^1,\mathbb{P}^1)$ which is independent with $W^0$, and we denote
 \begin{equation*}
 	\mu(X_t,t)=\mu_t,\;\nu(X_t,t)=\nu_t,\;\sigma(X_t,t)=\sigma_t,\;\pi(X_t,\mathcal{L}^1(X_t),t)=\pi_t.
 \end{equation*}
 As the previous $n$-agent game, we call $W^0$ the \textit{common noise} and $W$ an \textit{idiosyncratic noise}.\\
Next, we define the admissiblity set $\mathcal{A}_\text{MF}$ for this representative player to be the collection of $\mathcal{F}$-progressively measurable processes $(\pi_t)_{t\geq0}$, such that, 
\begin{equation*}
	\mathbb{E}\big[\int_0^t\vert\pi_s\mu_s\vert^2\mathrm{d}s\big]<\infty,\;\mathbb{E}\big[\int_0^t\vert\pi_s\nu_s\vert^4\mathrm{d}s\big]<\infty,\;\mathbb{E}\big[\int_0^t\vert\pi_s\sigma_s\vert^4\mathrm{d}s\big]<\infty,
\end{equation*}
for any $t>0$. We say that a portfolio strategy $\pi_t$ for the representative agent is \textit{admissible} if it belongs to $\mathcal{A}_\text{MF}$.\\

Since $\lim\limits_{n\to\infty}\sup_{t\geq0}\mathbb{E}\big[W_2(\overline{M}_t^n,\mathcal{L}^1(X_t))^2\big]=0$, $Y_t^n=h(\overline{M}_t^n)$ converges to $Y_t=h(\mathcal{L}^1(X_t))$, naturally. Indeed, convergence with respect to $W_2$ implies the weak convergence of measures. We can compute the dynamics of $Y_t$ by using the Itô formula along a flow of conditional measures. See \cite[Section 4.3]{CD2}. After some calculations, the dynamics of $Y_t$ is given by
\begin{equation*}
	\mathrm{d}Y_t=\mathbb{E}^1[\pi_t\mu_t]\mathrm{d}t+\mathbb{E}^1[\pi_t\sigma_t]\mathrm{d}W_t^0,\;\;\mathbb{P}^0\text{-a.s.}
\end{equation*}
We define a stochastic process $\alpha:\Omega\times[0,\infty)\to\mathcal{P}_2(\mathbb{R})$ by
\begin{equation*}
	\alpha_t(\omega^0)=\mathcal{L}\big(Y_t(\omega^0,\cdot)\big).
\end{equation*}

Note that the representative agent's utility is a random field $U:\Omega\times(0,\infty)\times\mathcal{P}_2(\mathbb{R})\times[0,\infty)\to\mathbb{R}$ such that,
\begin{equation}\label{eq:10}
	U(x,\mu,t)=-\exp\Big[-\frac{1}{\delta_t}\big(x-\theta_t\lambda(\mu)\big)+K_t\Big],
\end{equation}
where $K_t$ is an $\mathcal{F}$-progressively measurable process that is differentiable in time with $K_0=0$. Here the random parameters satisfy the conditions $\delta_t=\delta(X_t,t)>0,\;\theta_t=\theta(X_t,t)\in[0,1],\;\mathbb{P}$-almost surely, and they represent the representative agent's absolute risk tolerance and absolute competition weight.\\

We refer the concept of CARA-mean field (MF)-forward relative performance and CARA-MF-equilibrium in \cite{RP1}, and extend these to the measure dependent forward relative performance process.

\begin{definition} (CARA-MF-forward relative performance equilibrium for the generic agent)
	We assume that all agents have CARA risk preference. The $\mathbb{F}$-progressively measurable random field $U:\Omega\times(0,\infty)\times\mathcal{P}_2(\mathbb{R})\times[0,\infty)\to\mathbb{R}$ is an MF-forward relative performance for the generic agent if, for all $t\geq0$, the following conditions hold:
	\begin{enumerate}
		\item $U(x,\mu,t)$ is strictly increasing and strictly concave in $x$, $\mathbb{P}$-a.s.
		\item For each $\pi_t\in\mathcal{A}_\text{MF}$, $U(X_t^\pi,\alpha_t^\pi,t)$ is a $\mathbb{P}^1$-(local) supermartingale, where $X_t^\pi$ is the wealth process solving (\ref{eq:9}) for the strategy $\pi_t$, and $\alpha_t^\pi$ is the corresponding conditional marginal distribution.
		\item There exists $\pi_t^*\in\mathcal{A}_\text{MF}$ such that $U(X_t^*,\alpha_t^*,t)$ is a $\mathbb{P}^1$-(local) martingale, where $X_t^*$ is the wealth process solving (\ref{eq:9}) for the strategy $\pi_t^*$, and $\alpha_t^*$ is the corresponding conditional marginal distribution.
	\end{enumerate}
	We call the optimal strategy $\pi_t^*$ as an MF-equilibrium, and denote the pair $(U,\pi^*)$ satisfying above conditions as an MF-forward relative performance equilibrium.
\end{definition}

We now present the second main finding, in which we show the existence of the MF-forward relative performance equilibrium for the generic agent when this agent has a CARA risk preference.

\begin{theorem}\label{thm6}
	Define the random quantities
	\begin{equation*}
		\varphi^\sigma(t)=\mathbb{E}^1\Big[\frac{\mu_t\delta_t\sigma_t}{\Sigma_t}\Big],
	\end{equation*}
	\begin{equation*}
		\psi^\sigma(t)=\mathbb{E}^1\Big[\frac{\theta_t(\sigma_t)^2}{\Sigma_t}\Big],
	\end{equation*}
	and let
	\begin{equation}\label{eq:11}
		K_t=-\int_0^t\frac{\theta_s}{\delta_s}\mathbb{E}^1[\pi_s\mu_s]+\frac{1}{2}\Big(\frac{\theta_s}{\delta_s}\Big)^2\frac{(\nu_s)^2}{\Sigma_s}\mathbb{E}^1[\pi_s\sigma_s]^2-\frac{\theta_s}{2\delta_s}\frac{\mu_s\sigma_s}{\Sigma_s}\mathbb{E}^1[\pi_s\sigma_s]-\frac{(\mu_s)^2}{2\Sigma_s}\mathrm{d}s
	\end{equation}
  	where
  	\begin{equation*}
  		\mathbb{E}^1[\pi_t\sigma_t]=\frac{\varphi^\sigma(t)}{1-\psi^\sigma(t)},
  	\end{equation*}
  \begin{equation*}
  	\mathbb{E}^1[\pi_t\mu_t]=\mathbb{E}^1\Big[\frac{(\mu_t)^2\delta_t}{\Sigma_t}\Big]+\mathbb{E}^1\Big[\frac{\mu_t\theta_t\sigma_t}{\Sigma_t}\Big]\frac{\varphi^\sigma(t)}{1-\psi^\sigma(t)}.
  \end{equation*}
	 Furthermore, we assume that $U(x,\mu,t)$, utility of the representative agent, satisfies conditions (\ref{eq:2}), (\ref{eq:10}). 
	If $\psi^\sigma(t)\neq 1$ $\mathbb{P}^0$-a.s. for every $t\geq0$, then there exists an optimal strategy $\pi_t^{*}$ defined by
	\begin{equation*}
		\pi_t^{*}=\frac{1}{\Sigma_t}\Big(\mu_t\delta_t+\theta_t\sigma_t\frac{\varphi^\sigma(t)}{1-\psi^\sigma(t)}\Big),\;t\geq0.
	\end{equation*}
	Then $\big(U^{*},\pi^{*}\big)$ is an MF-forward relative performance equilibrium, where $U^{*}$ satisfies (\ref{eq:3}) with $X_t^{*},\;\alpha_t^{*}$ and $K_t$.
\end{theorem}

\begin{proof}[Proof]
	As in the $n$-agent game, we use the Itô-Wentzell-Lions formula in Proposition \ref{prop2} to compute the dynamics of a random field $U(x,\mu,t)$. Then we have $\mathbb{P}^0$-a.s.,
	\begin{align}\label{eq:12}
		\begin{split}
			\mathrm{d}U(X_t,\alpha_t,t)=\;&\frac{\partial U}{\partial t}(X_t,\alpha_t,t)
			\mathrm{d}t+\frac{\partial U}{\partial x}
			(X_t,\alpha_t,t)\pi_t\big(\mu_t\mathrm{d}t+\nu_t\mathrm{d}W_t+\sigma_t\mathrm{d}W_t^0\big)\\
			&+\frac{1}{2}\frac{\partial^2U}{\partial x^2}(X_t,\alpha_t,t)(\pi_t)^2\Sigma_t\mathrm{d}t\\
			&+\widetilde{\mathbb{E}}^1\bigg[\frac{\partial U}{\partial \mu}
			(X_t,\alpha_t,t)(\widetilde{Y_t})\mathbb{E}^1[\pi_t\mu_t]\bigg]\mathrm{d}t\\
			&+\widetilde{\mathbb{E}}^1\bigg[\frac{\partial U}{\partial \mu}
			(X_t,\alpha_t,t)(\widetilde{Y_t})\mathbb{E}^1[\pi_t\sigma_t]\bigg]\mathrm{d}W_t^0\\
			&+\frac{1}{2}\widetilde{\mathbb{E}}^1\bigg[\frac{\partial^2U}{\partial v \partial \mu}(X_t,\alpha_t,t)(\widetilde{Y_t})\big(\mathbb{E}^1[\pi_t\nu_t]^2+\mathbb{E}^1[\pi_t\sigma_t]^2\big)\bigg]\mathrm{d}t\\
			&+\frac{1}{2}\widehat{\mathbb{E}}^1\bigg[\widetilde{\mathbb{E}}^1\bigg[\frac{\partial^2U}{\partial \mu^2}(X_t,\alpha_t,t)(\widetilde{Y_t},\widehat{Y_t})\mathbb{E}^1[\pi_t\sigma_t]^2\bigg]\bigg]\mathrm{d}t\\
			&+\widetilde{\mathbb{E}}^1\bigg[\frac{\partial^2U}{\partial x \partial \mu}(X_t,\alpha_t,t)(\widetilde{Y_t})\pi_t\sigma_t\mathbb{E}^1[\pi_t\sigma_t]\bigg]\mathrm{d}t.
		\end{split}
	\end{align}
	Substituting the derivatives in Proposition \ref{prop3} into (\ref{eq:12}), we obtain $\mathbb{P}^0$-a.s.,
\begin{align}\label{eq:13}
	\begin{split}
		\frac{\mathrm{d}U(X_t,\alpha_t,t)}{U(X_t,\alpha_t,t)}=\;&\bigg[\frac{\Sigma_t}{2(\delta_t)^2}(\pi_t)^2-\frac{1}{(\delta_t)^2}\Big(\mu_t\delta_t+\theta_t\sigma_t\mathbb{E}^1[\pi_t\sigma_t]\Big)\pi_t\\
		&\qquad\qquad\qquad+\frac{\theta_t}{\delta_t}\Big(\mathbb{E}^1[\pi_t\mu_t]+\frac{\theta_t}{2\delta_t}\mathbb{E}^1[\pi_t\sigma_t]^2\Big)+\frac{\mathrm{d}K_t}{\mathrm{d}t}\bigg]\mathrm{d}t\\
		&-\frac{1}{\delta_t}\pi_t\nu_t\mathrm{d}W_t-\frac{1}{\delta_t}\Big(\pi_t\sigma_t-\theta_t\mathbb{E}^1[\pi_t\sigma_t]\Big)\mathrm{d}W_t^0.
	\end{split}
	&\end{align}
	For the process $U(x,\mu,t)$ to be an MF-forward relative performance, a value of the drift term in (\ref{eq:13}) should be less than or equal to 0 for all $\pi_t\in\mathcal{A}_\text{MF}$, and 0 for an optimal strategy. Note that the drift term in (\ref{eq:13}) is a quadratic function with respect to $\pi_t$ where the quadratic coefficient is negative.\\
	Let
	\begin{equation}\label{eq:14}
		\pi_t^*=\frac{1}{\Sigma_t}\Big(\mu_t\delta_t+\theta_t\sigma_t\mathbb{E}^1[\pi_t\sigma_t]\Big),
	\end{equation}
	 then the drift term is equal to $\frac{\Sigma_t^i}{2(\delta_t^i)^2}U^i\vert \pi_t^i-\pi_t^{i,*}\vert^2$ when $K_t$ is given as in (\ref{eq:11}).\\
	Now, multiplying both sides of (\ref{eq:14}) by $\sigma_t$, and taking conditional expectations on both sides, we obtain
	\begin{equation*}
		\mathbb{E}^1[\pi_t\sigma_t]=\frac{\varphi^\sigma(t)}{1-\psi^\sigma(t)}.
	\end{equation*}
	Substituting it into (\ref{eq:14}), we obtain a desired optimal strategy $\pi_t^{*}$.\\
	We can easily check that $\big(U^{*},\pi^{*}\big)$ satisfies conditions given in Definition \ref{def4}. Thus, $\big(U^{*},\pi^{*}\big)$ is an MF-forward relative performance equilibrium.\\
	Similarly, we have
	\begin{equation*}
		\mathbb{E}^1[\pi_t\mu_t]=\mathbb{E}^1\Big[\frac{(\mu_t)^2\delta_t}{\Sigma_t}\Big]+\mathbb{E}^1\Big[\frac{\mu_t\theta_t\sigma_t}{\Sigma_t}\Big]\frac{\varphi^\sigma(t)}{1-\psi^\sigma(t)}.
	\end{equation*}
\end{proof}

\begin{remark}
	As $n\to\infty$, the strategies $\pi_t^{i,*}$, weights $(\varphi_n^\sigma,\psi_n^\sigma)$, the utility $U^i(x,\mu,t)$ and an $\mathcal{F}$-progressively measurable process $K_t^i$ in Theorem \ref{thm5} converge to the corresponding quantities in Theorem \ref{thm6}.
\end{remark}

\section{CRRA risk preferences}\label{sec4}
We investigate agents who have power and logarithmic CRRA risk preferences with random individual relative risk tolerances and relative competition weights. Each agent measures its relative performance, taking into account competition among agents, and using the ratio between their own wealth and the average wealth of all agents as a benchmark. We proceed in a similar manner to 
Sect. \ref{sec3}, but the geometric average is used instead of the arithmetic average when computing the average wealth of all agents.
\subsection{The $n$-agent game}\label{subsec4.2}
We consider a game of $n$ agents competing with each other analogous to that of Subsect. \ref{subsec3.1}. We consider the identical market environment as in Sect. \ref{sec3}, that is, the dynamics of risky assets are the same as given in Subsect. \ref{subsec3.1}.\\

We suppose that each agent $i\in\{1,\ldots,n\}$ trades using a self-financing strategy $(\pi_t^i)_{t\geq0}$, which represents the (discounted by the bond) fraction(as opposed to the amount) of wealth invested in the $i$-th risky asset. Then dynamics of the $i$-th agent's wealth process $(X_t^i)_{t\geq0}$ is given by
\begin{equation}\label{eq:15}
	\mathrm{d}X_t^i=\pi_t^iX_t^i\big(\mu_t^i\mathrm{d}t+\nu_t^i\mathrm{d}W_t^i+\sigma_t^i\mathrm{d}W_t^0\big),\;\;X_0^i=x_0^i\in\mathbb{R},
\end{equation}
where
\begin{equation*}
	\mu_t^i=\mu(X_t^i,t),\;\nu_t^i=\nu(X_t^i,t),\;\sigma_t^i=\sigma(X_t^i,t),\;\pi_t^i=\pi\big(X_t^i,\frac{1}{n}\sum_{k=1}^n \delta_{X_t^k},t\big).
\end{equation*}
Since we consider the power function as an utility in this Section, it is natural to consider the $\pi_t^i$ as a fraction of wealth.\\
Next, we define the admissibility set $\mathcal{A}^i$ for the agent $i$ analogous to Subsect. \ref{subsec3.1}, that is the collection of $\mathcal{F}$-progressively measurable processes $(\pi_t^i)_{t\geq0}$, such that, 
\begin{equation*}
	\mathbb{E}\big[\int_0^t\vert\pi_s^i\mu_s^iX_s^i\vert^2\mathrm{d}s\big]<\infty,\;\mathbb{E}\big[\int_0^t\vert\pi_s^i\nu_s^iX_s^i\vert^4\mathrm{d}s\big]<\infty,\;\mathbb{E}\big[\int_0^t\vert\pi_s^i\sigma_s^iX_s^i\vert^4\mathrm{d}s\big]<\infty,
\end{equation*}
for any $t>0$. We say that a portfolio strategy $\pi_t^i$ for the agent $i$ is \textit{admissible} if it belongs to $\mathcal{A}^i$.\\

Considering competition among agents in order to measure a relative performance, we introduce a stochastic process $(Y_t^n)_{t\geq0}$ by
\begin{equation*}
	Y_t^n=h(\overline{M}_t^n),
\end{equation*} where
\begin{equation*}
	\overline{M}_t^n=\frac{1}{n}\sum_{k=1}^n \delta_{X_t^k},
\end{equation*} 
and
\begin{equation*}
	h(\mu)=\exp\Big(\int_\mathbb{R} \log x\;\mathrm{d}\mu(x)\Big).
\end{equation*}
Then $Y_t^n=\Big(\prod_{k=1}^nX_t^k\Big)^{1/n}$, and it is the geometric average wealth of all agents. We can compute the dynamics of $Y_t^n$ in two ways. Firstly, we can derive it from the general method $\mathrm{d}Y_t^n=\mathrm{d}\Big(\prod_{k=1}^nX_t^k\Big)^{1/n}$, by using the logarithm $\mathrm{d}\log(Y_t)$. Secondly, we can use the empirical projection $h^{(n)}(X_t^1,\ldots,X_t^n)$ and derive it from $\mathrm{d}Y_t^n=\mathrm{d}h^{(n)}(X_t^1,\ldots,X_t^n)$. After some calculations using Proposition \ref{prop1}, the dynamics of $Y_t^i$ is given by
\begin{equation*}
	\frac{\mathrm{d}Y_t^n}{Y_t^n}=\eta_t^n\mathrm{d}t+\frac{1}{n}\sum_{k=1}^n\pi_t^k\nu_t^k\mathrm{d}W_t^k+(\overline{\pi\sigma})_t\mathrm{d}W_t^0,
\end{equation*}
where
\begin{equation*}
	\eta_t^n=(\overline{\pi\mu})_t+\frac{1}{2}\Big((\overline{\pi\sigma})_t^2+\frac{1}{n}(\overline{(\pi\nu)^2})_t-(\overline{\pi^2\Sigma})_t\Big),
\end{equation*}
\begin{equation*}
	(\overline{(\pi\nu)^2})_t=\frac{1}{n}\sum_{k=1}^n(\pi_t^k\nu_t^k)^2,
\end{equation*}
\begin{equation*}
	(\overline{\pi^2\Sigma})_t=\frac{1}{n}\sum_{k=1}^n(\pi_t^k)^2\Sigma_t^k.
\end{equation*}
We define a stochastic process $\alpha^n:\Omega^0\times[0,\infty)\to\mathcal{P}_2(\mathbb{R})$ by
\begin{equation*}
	\alpha_t^n(\omega^0)=\mathcal{L}\big(Y_t^n(\omega^0,\cdot)\big).
\end{equation*}

\begin{remark}
	Similary to \cite[Remark 3.3]{LZ1}, it is more natural to replace the average wealth $Y_t^n$ with the average over all other agents. If we define $Y_t^i$ by
	\begin{equation*}
		Y_t^i=p_n(\overline{M}_t^n)q_n(X_t^i)
	\end{equation*}
	where
	\begin{equation*}
		p_n(\mu)=\exp\Big(\int_\mathbb{R} \frac{n}{n-1}\log x\;\mathrm{d}\mu(x)\Big)\;
	\end{equation*}
	and
	\begin{equation*}
		q_n(x)=x^{-\frac{1}{n-1}},
	\end{equation*}
	then we can obtain the corresponding results for that case.
\end{remark}

We assume that the $i$-th agent's utility is a random field $U^i:\Omega\times(0,\infty)\times\mathcal{P}_2(\mathbb{R})\times[0,\infty)\to\mathbb{R}$ such that
\begin{equation}\label{eq:16}
	U^i(x,\mu,t)=\begin{cases}
		\big(1-\frac{1}{\delta_t^i}\big)^{-1}\Big(\lambda(\mu)^{-\theta_t^i}x\Big)^{1-\frac{1}{\delta_t^i}}K_t^i, & \text{if }\delta_t^i\neq1\\
			\log\Big(\lambda(\mu)^{-\theta_t^i}x\Big)K_t^i+G_t^i, & \text{if }\delta_t^i=1
		\end{cases},
\end{equation}
where $K_t^i$ and $G_t^i$ are $\mathcal{F}$-progressively measurable processes that are differentiable in time with $K_0^i=1$ and $G_0^i=0$, and we define
\begin{equation*}
	\lambda(\mu)=\exp\Big(\int_\mathbb{R} \log x\;\mathrm{d}\mu(x)\Big).
\end{equation*}
Here, the random parameters satisfy the conditions $\delta_t^i=\delta(X_t^i,t)>0$, $\theta_t^i=\theta(X_t^i,t)\in[0,1]$, $\mathbb{P}$-almost surely, and they represent the $i$-th agent's relative risk tolerance and relative competition weight. Note that
\begin{equation*}
	U^i(X_t^i,\alpha_t^n,t)=\begin{cases}
		\big(1-\frac{1}{\delta_t^i}\big)^{-1}\Big(X_t^i\overline{X}_t^{-\theta_t^i}\Big)^{1-\frac{1}{\delta_t^i}}K_t^i, &\text{if }\delta_t^i\neq1\\
		\log\Big(X_t^i\overline{X}_t^{-\theta_t^i}\Big)K_t^i+G_t^i, & \text{if }\delta_t^i=1
		\end{cases}
\end{equation*}
where
\begin{equation*}
	\overline{X}_t=\Big(\prod_{k=1}^nX_t^k\Big)^{1/n},
\end{equation*}
which is almost the same as the CRRA power and logarithmic utility forms used in \cite{LS1,LZ1,RP3}. Note that
\begin{equation*}
	X_t^i\overline{X}_t^{-\theta_t^i}=(X_t^i)^{1-\theta_t^i}\Big(\frac{X_t^i}{\overline{X}_t}\Big)^{\theta_t^i}.
\end{equation*}
Thus, the smaller value of $\theta_t^i$, the relative performance becomes less relevant.\\

We can check that partial derivatives and L-derivatives of $U^i$ exist and they are given in the following Proposition.
\begin{proposition}\label{prop7}
	For each agent $i\in\{1,\ldots,n\}$, the utility $U^i$ has the partial derivatives and L-derivatives given as follows:
	\begin{equation*}
		\frac{\partial U^i}{\partial t}(x,\mu,t)=\begin{cases}
			\frac{1}{K_t^i}\frac{\mathrm{d}K_t^i}{\mathrm{d}t}U^i(x,\mu,t), &\text{if }\delta_t^i\neq1\\
			\log\Big(\lambda(\mu)^{-\theta_t^i}x\Big)\frac{\mathrm{d}K_t^i}{\mathrm{d}t}+\frac{\mathrm{d}G_t^i}{\mathrm{d}t}, & \text{if }\delta_t^i=1
			\end{cases},
	\end{equation*}
\begin{equation*}
	\frac{\partial U^i}{\partial x}(x,\mu,t)=\begin{cases}
		\big(1-\frac{1}{\delta_t^i}\big)x^{-1}U^i(x,\mu,t), &\text{if }\delta_t^i\neq1\\
		x^{-1}K_t^i, & \text{if }\delta_t^i=1
	\end{cases},
\end{equation*}
\begin{equation*}
	\frac{\partial^2U^i}{\partial x^2}(x,\mu,t)=\begin{cases}
		-\frac{1}{\delta_t^i}\big(1-\frac{1}{\delta_t^i}\big)x^{-2}U^i(x,\mu,t), &\text{if }\delta_t^i\neq1\\
		-x^{-2}K_t^i, & \text{if }\delta_t^i=1
	\end{cases},
\end{equation*}
\begin{equation*}
	\frac{\partial U^i}{\partial \mu}(x,\mu,t)(v)=\begin{cases}
		-\big(1-\frac{1}{\delta_t^i}\big)\theta_t^i\frac{1}{v}U^i(x,\mu,t), &\text{if }\delta_t^i\neq1\\
		-\theta_t^i\frac{1}{v}K_t^i, & \text{ if }\delta_t^i=1
	\end{cases},
\end{equation*}
\begin{equation*}
	\frac{\partial^2U^i}{\partial v \partial \mu}(x,\mu,t)(v)=\begin{cases}
		\big(1-\frac{1}{\delta_t^i}\big)\theta_t^i\frac{1}{v^2}U^i(x,\mu,t),&\text{if }\delta_t^i\neq1\\
		\theta_t^i\frac{1}{v^2}K_t^i ,&\text{if }\delta_t^i=1
	\end{cases},
\end{equation*}
\begin{equation*}
	\frac{\partial^2U^i}{\partial \mu^2}(x,\mu,t)(v,v')=\begin{cases}
		\big(1-\frac{1}{\delta_t^i}\big)^2(\theta_t^i)^2\frac{1}{v}\frac{1}{v'}U^i(x,\mu,t) ,&\text{if }\delta_t^i\neq1\\
		\quad0 ,& \text{if }\delta_t^i=1
	\end{cases},
\end{equation*}
\begin{equation*}
	\frac{\partial^2U^i}{\partial x \partial \mu}(x,\mu,t)(v)=\begin{cases}
		-\big(1-\frac{1}{\delta_t^i}\big)^2\theta_t^i\frac{1}{v}x^{-1}U^i(x,\mu,t) ,&\text{if }\delta_t^i\neq1\\
		\quad0 ,& \text{if }\delta_t^i=1
	\end{cases}.
\end{equation*}
\end{proposition}
\begin{proof}[Proof]
	It is known that the L-derivative of $u(\mu)$ of the form $\int_{\mathbb{R}}f(x)\mathrm{d}\mu(x)$ is given by $\frac{\partial f}{\partial x}$ if $\frac{\partial f}{\partial x}(x)$ is at most of linear growth. Since $\frac{\mathrm{d}(\log v)}{\mathrm{d} v}=\frac{1}{v}$ is at most of linear growth, we obtain $\partial_\mu \lambda(\mu)(v)=\frac{1}{v}\lambda(\mu)$. Other derivatives can be obtained using the usual differentiation method.
\end{proof}

We now define the CRRA-\textit{forward relative performance} analogous to previous Section.

\begin{definition}(CRRA-forward relative performance for the agent)\\
	We assume that all agents have CRRA risk preferences. For each agent $i$, suppose that a strategy for the other agent $j$, $\pi_t^j$, is arbitrary but fixed, for all $j\neq i$. The $\mathcal{F}$-progressively measurable random field $U^i:\Omega\times(0,\infty)\times\mathcal{P}_2(\mathbb{R})\times[0,\infty)\to\mathbb{R}$ is a forward relative performance for the agent $i$ if, for all $t\geq0$, it satisfies:
	\begin{enumerate}
		\item $U^i(x,\mu,t)$ is $\mathbb{P}$-a.s. strictly increasing and strictly concave in $x$.
		\item For any $\pi_t^i\in\mathcal{A}^i$, $U^i(X_t^i,\alpha_t^n,t)$ is a $\mathbb{P}^1$-(local) supermartingale $\mathbb{P}^0$-a.s.
		\item There exists $\pi_t^{i,*}\in\mathcal{A}^i$, such that $U^i(X_t^{i,*},\alpha_t^n,t)$ is a $\mathbb{P}^1$-(local) martingale $\mathbb{P}^0$-a.s. where $X_t^{i,*}$ satisfies (\ref{eq:15}) with a strategy $\pi_t^{i,*}$.
	\end{enumerate}
\end{definition}

When the other agents' strategies are arbitrary but fixed, we can reply to this strategy to obtain the agent's own optimal investment policy. Here we derive the random ODEs which $K_t^i$ and $G_t^i$ satisfy in order for $U^i$ to be a forward relative performance, and investigate the optimal strategy for each agent.

\begin{proposition}\label{prop8}(Best responses)\\
	Fix $i\in\{1,\ldots,n\}$ and suppose that $\pi_t^j$ for the other agents $j\neq i$ are arbitrary but fixed. We assume that $U^i$, the utility of agent $i$, satisfies conditions (\ref{eq:2}), (\ref{eq:16}) for given $K_t^i$ and $G_t^i$. Consider the random ODEs, in the case that $\delta_t^i\neq1$,
	\begin{align}\label{eq:17}
		\begin{split}
		\frac{\mathrm{d}K_t^i}{\mathrm{d}t}=\Big(1-\frac{1}{\delta_t^i}\Big)\bigg[\theta_t^i\Big(\mathbb{E}^1&[(\overline{\pi\mu})_t]-\frac{1}{2}\mathbb{E}^1[(\overline{\pi^2\Sigma})_t]\Big)-\Big(1-\frac{1}{\delta_t^i}\Big)(\theta_t^i)^2\mathbb{E}^1[(\overline{\pi\sigma})_t]^2\\
		&-\frac{1}{2\delta_t^i\Sigma_t^i}\Big(\mu_t^i\delta_t^i+(1-\delta_t^i)\theta_t^i\sigma_t^i\mathbb{E}^1[(\overline{\pi\sigma})_t]\Big)^2\bigg]K_t^i,
		\end{split}
	\end{align}
	and in the case that $\delta_t^i=1$,
	\begin{equation}\label{eq:18}
		\log\Big(\lambda(\alpha_t^n)^{-\theta_t^i}X_t^i\Big)\frac{\mathrm{d}K_t^i}{\mathrm{d}t}+\frac{\mathrm{d}G_t^i}{\mathrm{d}t}=\Big(\theta_t^i\Big(\mathbb{E}^1[(\overline{\pi\mu})_t]-\frac{1}{2}\mathbb{E}^1[(\overline{\pi^2\Sigma})_t]\Big)-\frac{(\mu_t^i)^2}{2\Sigma_t^i}\Big)K_t^i.
	\end{equation}
	Define the strategy $\pi_t^{i,*}$ by
	\begin{equation}\label{eq:19}
		\pi_t^{i,*}=\frac{1}{\Sigma_t^i}\Big(\mu_t^i\delta_t^i+(1-\delta_t^i)\theta_t^i\sigma_t^i\mathbb{E}^1\big[(\overline{\pi\sigma})_t\big]\Big),\;\; t\geq0.
	\end{equation}
	In the case that $\delta_t^i\neq 1$, if $K_t^i$ satisfies (\ref{eq:17}) $\mathbb{P}^0$-a.s., then $U^i(x,\mu,t)$ is a forward relative performance process, and the policy $\pi_t^{i,*}$ is an optimal strategy for an agent $i$.\\
	In the case that $\delta_t^i= 1$, if $K_t^i$ and $G_t^i$ satisfies (\ref{eq:18}) $\mathbb{P}^0$-a.s., then $U^i(x,\mu,t)$ is a forward relative performance process, and the policy $\pi_t^{i,*}$ is an optimal strategy for an agent $i$.
\end{proposition}

\begin{proof}[Proof.]
	As in Proposition \ref{prop4}, we use the Itô-Wentzell-Lions formula in Proposition \ref{prop2} to compute the dynamics of a random field $U^i(x,\mu,t)$. Then we have $\mathbb{P}^0$-a.s.,
	\begin{align}\label{eq:20}
		\begin{split}
			\mathrm{d}U^i(X_t^i,\alpha_t^n,t)=\;&\frac{\partial U^i}{\partial t}(X_t^i,\alpha_t^n,t)\mathrm{d}t\\
			&+\frac{\partial U^i}{\partial x}(X_t^i,\alpha_t^n,t)\pi_t^iX_t^i\big(\mu_t^i\mathrm{d}t+\nu_t^i\mathrm{d}W_t^i+\sigma_t^i\mathrm{d}W_t^0\big)\\
			&+\frac{1}{2}\frac{\partial^2U^i}{\partial x^2}(X_t^i,\alpha_t^n,t)(\pi_t^i)^2(X_t^i)^2\Sigma_t^i\mathrm{d}t\\
			&+\widetilde{\mathbb{E}}^1\bigg[\frac{\partial U^i}{\partial \mu}(X_t^i,\alpha_t^n,t)(\widetilde{Y}_t^n)\frac{\widetilde{\eta}_t^n}{\widetilde{Y}_t^n}\bigg]\mathrm{d}t\\
			&+\widetilde{\mathbb{E}}^1\bigg[\frac{\partial U^i}{\partial \mu}(X_t^i,\alpha_t^n,t)(\widetilde{Y}_t^n)\frac{\overline{(\widetilde{\pi\sigma})}_t}{\widetilde{Y}_t^n}\bigg]\mathrm{d}W_t^0\\
			&+\frac{1}{2}\widetilde{\mathbb{E}}^1\bigg[\frac{\partial^2U^i}{\partial v \partial \mu} U^i(X_t^i,\alpha_t^n,t)(\widetilde{Y}_t^n)\bigg(\frac{\frac{1}{n}\overline{\widetilde{(\pi\nu
						)^2}}_t}{(\widetilde{Y}_t^n)^2}+\Big(\frac{\overline{(\widetilde{\pi\sigma})}_t}{\widetilde{Y}_t^n}\Big)^2\bigg)\bigg]\mathrm{d}t\\
			&+\frac{1}{2}\widehat{\mathbb{E}}^1\bigg[\widetilde{\mathbb{E}}^1\bigg[\frac{\partial^2U^i}{\partial \mu^2}(X_t^i,\alpha_t^n,t)(\widetilde{Y}_t^n,\widehat{Y}_t^n)\frac{\overline{(\widetilde{\pi\sigma})}_t}{\widetilde{Y}_t^n}\frac{\overline{(\widehat{\pi\sigma})}_t}{\widehat{Y}_t^n}\bigg]\bigg]\mathrm{d}t\\
			&+\widetilde{\mathbb{E}}^1\bigg[\frac{\partial^2U^i}{\partial x \partial \mu}(X_t^i,\alpha_t^n,t)(\widetilde{Y}_t^n)\pi_t^iX_t^i\sigma_t^i\frac{\overline{(\widetilde{\pi\sigma})}_t}{\widetilde{Y}_t^n}\bigg]\mathrm{d}t.
		\end{split}
	\end{align}
	Substituting the derivatives in Proposition \ref{prop7} into (\ref{eq:20}), we obtain $\mathbb{P}^0$-a.s., in case that $\delta_t^i\neq1$,
\begin{align}\label{eq:21}
	\begin{split}
		\frac{\mathrm{d}U^i(X_t^i,\alpha_t^n,t)}{U^i(X_t^i,\alpha_t^n,t)}=\;&\Big(1-\frac{1}{\delta_t^i}\Big)\bigg[-\frac{\Sigma_t^i}{2\delta_t^i}(\pi_t^i)^2+\frac{1}{\delta_t^i}\Big(\mu_t^i\delta_t^i+(1-\delta_t^i)\theta_t^i\sigma_t^i\mathbb{E}^1[(\overline{\pi\sigma})_t]\Big)\pi_t^i\\
		&\qquad\qquad \qquad+\Big(1-\frac{1}{\delta_t^i}\Big)(\theta_t^i)^2\mathbb{E}^1[(\overline{\pi\sigma})_t]^2 \\
		&\qquad\qquad \qquad-\theta_t^i\Big(\mathbb{E}^1[(\overline{\pi\mu})_t]-\mathbb{E}^1[(\overline{\pi^2\Sigma})_t]\Big)+\frac{1}{K_t^i}\frac{\mathrm{d}K_t^i}{\mathrm{d}t} \bigg]\mathrm{d}t\\
		&+\Big(1-\frac{1}{\delta_t^i}\Big)\bigg(\pi_t^i\nu_t^i\mathrm{d}W_t^i+\Big(\pi_t^i\sigma_t^i-\theta_t^i\mathbb{E}^1[(\overline{\pi\sigma})_t]\Big)\mathrm{d}W_t^0\bigg).
	\end{split}
\end{align}
	For the process $U^i(x,\mu,t)$ to be a forward relative performance, a value of the drift term in (\ref{eq:21}) should be less than or equal to 0 for all $\pi_t^i\in\mathcal{A}^i$, and 0 for an optimal strategy. Note that the drift term in (\ref{eq:21}) is a quadratic function with respect to $\pi_t^i$ where the quadratic coefficient is negative.\\
	If $K_t^i$ satisfies (\ref{eq:17}) and we define $\pi_t^{i,*}$ as in (\ref{eq:19}), then we can simplify the drift term to $-\frac{1}{2}\frac{\Sigma_t^i}{\delta_t^i}(X_t^i)^{1-\frac{1}{\delta_t^i}}\vert \pi_t^i-\pi_t^{i,*}\vert^2$. Thus, we conclude that $\pi_t^{i,^*}$ is an optimal strategy and $U^i(x,\mu,t)$ is a forward relative performance for agent $i$.\\
	In case that $\delta_t^i=1$, we can obtain analogous results, $\mathbb{P}^0$-a.s. we have,
	\begin{align}\label{eq:22}
		\begin{split}
			\mathrm{d}U^i(X_t^i,\alpha_t^n,t)=\;&\bigg[-\frac{\Sigma_t^i}{2}K_t^i(\pi_t^i)^2+\mu_t^i\pi_t^i-\theta_t^iK_t^i\Big(\mathbb{E}^1[(\overline{\pi\mu})_t]-\frac{1}{2}\mathbb{E}^1[(\overline{\pi^2\Sigma})_t]\Big)\\
			&\qquad\qquad\qquad\qquad\qquad\quad+\Big(\log\Big(\lambda(\alpha_t^n)^{-\theta_t^i}X_t^i\Big)\frac{\mathrm{d}K_t^i}{\mathrm{d}t}+\frac{\mathrm{d}G_t^i}{\mathrm{d}t}\Big)\bigg]\mathrm{d}t\\
			&+K_t^i\pi_t^i\Big(\nu_t^i\mathrm{d}W_t^i+\sigma_t^i\mathrm{d}W_t^0\Big).
		\end{split}
	\end{align}
	Thus, through a similar procedure, if $K_t^i$ and $G_t^i$ satisfy (\ref{eq:18}) and we define $\pi_t^{i,*}$ as in (\ref{eq:19}), then we conclude that $\pi_t^{i,*}$ is an optimal strategy and $U^i(x,\mu,t)$ is a forward relative performance for agent $i$.
\end{proof}

We define the CRRA-\textit{forward Nash equilibrium} analogous to Subsect. \ref{subsec3.1}.

\begin{definition}(CRRA-forward Nash equilibrium)\label{def7}
We assume that each agent $i$ has CRRA risk preference $U^i:\Omega\times(0,\infty)\times\mathcal{P}_2(\mathbb{R})\times[0,\infty)\to\mathbb{R}$, $1\leq i\leq n$. A forward Nash equilibrium consists of $n$-pairs of $\mathcal{F}$-progressively measurable random fields $\big(U^i,\pi^{i,*}\big)$, $1\leq i\leq n$, such that, for any $t\geq0$, the following conditions hold:
	\begin{enumerate}
		\item For each $i$, $\pi_t^{i,*}\in\mathcal{A}^i$.
		\item For each $i$, $U^i(x,\mu,t)$ is $\mathbb{P}$-a.s. strictly increasing and strictly concave in $x$.
		\item Let agents $j\neq i$ invest along the strategy $\pi_t^{j,*}$ and agent $i$ invests along the strategy $\pi_t^i\in\mathcal{A}^i$. Then $U^i(X_t^i,\alpha_t^n,t)$ is a $\mathbb{P}^1$-(local) supermartingale $\mathbb{P}^0$-a.s.
		\item Let all agents $i$ invest along the strategy $\pi_t^{i,*}$ and $X_t^{i,*}$ be the associated wealth process and $\alpha_t^{n,*}$ be the associated conditional marginal distribution. Then $U^i(X_t^{i,*},\alpha_t^{n,*},t)$ is a $\mathbb{P}^1$-(local) martingale $\mathbb{P}^0$-a.s.
	\end{enumerate}
\end{definition}

We now present the third main result, in which we find the forward Nash equilibrium when all agents have CRRA risk preferences.

\begin{theorem}\label{thm9}
	Assume the conditions of Proposition \ref{prop8} hold for all agents $i\in\{1,\ldots,n\}$. Define the random quantities
	\begin{equation*}
		\varphi_n^\sigma(t)=\mathbb{E}^1\Big[\frac{1}{n}\sum_{k=1}^n\frac{\mu_t^k\delta_t^k\sigma_t^k}{\Sigma_t^k}\Big],
	\end{equation*}
	\begin{equation*}
		\psi_n^\sigma(t)=\mathbb{E}^1\Big[\frac{1}{n}\sum_{k=1}^n(1-\delta_t^k)\frac{\theta_t^k(\sigma_t^k)^2}{\Sigma_t^k}\Big].
	\end{equation*}
	If $\psi_n^\sigma(t)\neq 1$ $\mathbb{P}^0$-a.s. for every $t\geq0$, then there exists an optimal strategy $\pi_t^{i,*}$ defined by
	\begin{equation*}
		\pi_t^{i,*}=\frac{1}{\Sigma_t^i}\Big(\mu_t^i\delta_t^i+(1-\delta_t^i)\theta_t^i\sigma_t^i\frac{\varphi_n^\sigma(t)}{1-\psi_n^\sigma(t)}\Big),\;t\geq0.
	\end{equation*}
	Furthermore, let
	\begin{equation}\label{eq:23}
		K_t^i=\begin{cases}
			-\exp\bigg[\displaystyle\int_0^t\Big(1-\frac{1}{\delta_s^i}\Big)\bigg(-\theta_s^i\Big(\mathbb{E}^1[(\overline{\pi\mu})_s]-\frac{1}{2}\mathbb{E}^1[(\overline{\pi^2\Sigma})_s]\Big)&\\
			\qquad\qquad+\Big(1-\frac{1}{\delta_s^i}\Big)(\theta_s^i)^2\mathbb{E}^1[(\overline{\pi\sigma})_s]^2&\\
			\qquad\qquad+\frac{1}{2\delta_s^i\Sigma_s^i}\Big(\mu_s^i\delta_s^i+(1-\delta_s^i)\theta_s^i\sigma_s^i\mathbb{E}^1[(\overline{\pi\sigma})_s]\Big)^2\bigg)\mathrm{d}s\bigg], &\text{if } \delta_t^i\neq1\\
			\qquad1 ,&\text {if } \delta_t^i=1
			\end{cases},
	\end{equation}
	and
	\begin{equation}\label{eq:24}
		G_t^i=-\displaystyle\int_0^t\Big(-\theta_s^i\Big(\mathbb{E}^1[(\overline{\pi\mu})_s]-\frac{1}{2}\mathbb{E}^1[(\overline{\pi^2\Sigma})_t]\Big)+\frac{(\mu_s^i)^2}{2\Sigma_s^i}\Big)\mathrm{d}s,
	\end{equation}
	where
	\begin{equation*}
		\mathbb{E}^1[(\overline{\pi\sigma})_t]=\frac{\varphi_n^\sigma(t)}{1-\psi_n^\sigma(t)},
	\end{equation*}
	\begin{equation*}
		\mathbb{E}^1[(\overline{\pi\mu})_t]=\mathbb{E}^1\Big[\frac{1}{n}\sum_{k=1}^n\frac{(\mu_t^k)^2\delta_t^k}{\Sigma_t^k}\Big]+\mathbb{E}^1\Big[\frac{1}{n}\sum_{k=1}^n(1-\delta_t^k)\frac{\mu_t^k\theta_t^k\sigma_t^k}{\Sigma_t^k}\Big]\frac{\varphi_n^\sigma(t)}{1-\psi_n^\sigma(t)},
	\end{equation*}
	\begin{equation*}
		\mathbb{E}^1[(\overline{\pi^2\Sigma})_t]=\mathbb{E}^1\bigg[\frac{1}{n}\sum_{k=1}^n\frac{1}{\Sigma_t^k}\bigg(\mu_t^k\delta_t^k+(1-\delta_t^k)\theta_t^k\sigma_t^k\frac{\varphi_n^\sigma(t)}{1-\psi_n^\sigma(t)}\bigg)^2\bigg].
	\end{equation*}
	 Then $n$-pairs of $\big(U^{i,*},\pi^{i,*}\big)$ are a forward Nash equilibrium, where $U^{i,*}$ satisfies (\ref{eq:16}) with $X_t^{i,*},\;\alpha_t^{n,*},\;K_t^i,$ and $G_t^i$.
\end{theorem}

\begin{proof}[Proof]
	Multiplying both sides of (\ref{eq:19}) by $\sigma_t^i$, and averaging over $i$'s, then
	\begin{equation*}
		(\overline{\pi\sigma})_t=\frac{1}{n}\sum_{k=1}^n\frac{\mu_t^k\delta_t^k\sigma_t^k}{\Sigma_t^k}+\frac{1}{n}\sum_{k=1}^n(1-\delta_t^k)\frac{\theta_t^k(\sigma_t^k)^2}{\Sigma_t^k}\mathbb{E}^1[(\overline{\pi\sigma})_t].
	\end{equation*}
	Now, taking conditional expectations on both sides, we obtain
	\begin{equation*}
		\mathbb{E}^1[(\overline{\pi\sigma})_t]=\frac{\varphi_n^\sigma(t)}{1-\psi_n^\sigma(t)}.
	\end{equation*}
	Substituting it into (\ref{eq:19}), we obtain a desired optimal strategy $\pi_t^{i,*}$.\\
	Furthermore, if we let $K_t^i$ and $G_t^i$ be given by (\ref{eq:23}), (\ref{eq:24}), then it satisfies (\ref{eq:17}) for $\delta_t^i\neq1$, and (\ref{eq:18}) for $\delta_t^i=1$.
	Thus, $U^i(x,\mu,t)$ is a forward relative performance process. We can easily check that $\big(U^{i,*},\pi^{i,*}\big)$ satisfies conditions given in Definition \ref{def7}. Thus, $n$-pairs of $\big(U^{i,*},\pi^{i,*}\big)$ are a forward Nash equilibrium.\\
	Similarly, we have
	\begin{equation*}
		\mathbb{E}^1[(\overline{\pi\mu})_t]=\mathbb{E}^1\Big[\frac{1}{n}\sum_{k=1}^n\frac{(\mu_t^k)^2\delta_t^k}{\Sigma_t^k}\Big]+\mathbb{E}^1\Big[\frac{1}{n}\sum_{k=1}^n(1-\delta_t^k)\frac{\mu_t^k\theta_t^k\sigma_t^k}{\Sigma_t^k}\Big]\frac{\varphi_n^\sigma(t)}{1-\psi_n^\sigma(t)},
	\end{equation*}
	\begin{equation*}
		\mathbb{E}^1[(\overline{\pi^2\Sigma})_t]=\mathbb{E}^1\bigg[\frac{1}{n}\sum_{k=1}^n\frac{1}{\Sigma_t^k}\bigg(\mu_t^k\delta_t^k+(1-\delta_t^k)\theta_t^k\sigma_t^k\frac{\varphi_n^\sigma(t)}{1-\psi_n^\sigma(t)}\bigg)^2\bigg].
	\end{equation*}
\end{proof}

\subsection{The mean field stochastic optimal control}\label{subsec4.2}

In this Section, we study the limit as $n\to\infty$ of the $n$-agent game analyzed in previous Section, analogously to the CARA exponential risk preferences in Subsect. \ref{subsec3.2}\\

Recall that we perform passaging to the limit $n\to\infty$ first to investigate the asymptotic regime, and then optimize over controlled dynamics of McKean-Vlasove type. We define the stochastic processes $(\underline{X}_t^i)_{t\geq0}$, for $i\in\{1,\ldots,n\}$, satisfy the following McKean-Vlasov SDE,
\begin{equation*}
	\mathrm{d}\underline{X}_t^i=\pi\big(\underline{X}_t^i,\mathcal{L}^1(\underline{X}_t^i),t\big)\underline{X}_t^i\Big[\mu(\underline{X}_t^i,t)\mathrm{d}t+\nu(\underline{X}_t^i,t)\mathrm{d}W_t^i+\sigma(\underline{X}_t^i,t)\mathrm{d}W_t^0\Big].
\end{equation*}
Then the theory of conditional propagation of chaos states that, for all $i\geq1$, the followings hold:
\begin{equation*}
	\mathbb{P}^0\big\{\mathcal{L}^1(\underline{X}_t^i)=\mathcal{L}^1(\underline{X}_t^1),\;\text{for any}\;t\geq0\big\}=1,
\end{equation*}
\begin{equation*}
	\lim_{n\to\infty}\sup_{t\geq0}\mathbb{E}\big[W_2(\overline{M}_t^n,\mathcal{L}^1(\underline{X}_t^1))^2\big]=0.
\end{equation*}
Therefore, we can consider a representative player $X=(X_t)_{t\geq0}$ such that
\begin{equation*}
	\mathbb{P}^0\big\{\mathcal{L}^1(X_t)=\mathcal{L}^1(\underline{X}_t^1),\;\text{for any}\;t\geq0\big\}=1,
\end{equation*}
and its dynamics is given by
\begin{equation}\label{eq:25}
	\mathrm{d}X_t=\pi_tX_t\big(\mu_t\mathrm{d}t+\nu_t\mathrm{d}W_t+\sigma_t\mathrm{d}W_t^0\big),
\end{equation}
where
\begin{equation*}
	\mu(X_t,t)=\mu_t,\;\nu(X_t,t)=\nu_t,\;\sigma(X_t,t)=\sigma_t,\;\pi(X_t,\mathcal{L}^1(X_t),t)=\pi_t.
\end{equation*}
Next, we define the admissiblity set $\mathcal{A}_\text{MF}$ for this representative player to be the collection of $\mathcal{F}$-progressively measurable processes $(\pi_t)_{t\geq0}$, such that, 
\begin{equation*}
	\mathbb{E}\big[\int_0^t\vert\pi_s\mu_sX_s\vert^2\mathrm{d}s\big]<\infty,\;\mathbb{E}\big[\int_0^t\vert\pi_s\nu_sX_s\vert^4\mathrm{d}s\big]<\infty,\;\mathbb{E}\big[\int_0^t\vert\pi_s\sigma_sX_s\vert^4\mathrm{d}s\big]<\infty,
\end{equation*}
for any $t>0$. We say that a portfolio strategy $\pi_t$ for the representative agent is \textit{admissible} if it belongs to $\mathcal{A}_\text{MF}$.\\

Since $\lim\limits_{n\to\infty}\sup_{t\geq0}\mathbb{E}\big[W_2(\overline{M}_t^n,\mathcal{L}^1(X_t))^2\big]=0$, $Y_t^n=h(\overline{M}_t^n)$ converges to $Y_t=h(\mathcal{L}^1(X_t))$, naturally. We can compute the dynamics of $Y_t$ by using the Itô formula along a flow of conditional measures. See \cite[Section 4.3]{CD2}. After some calculations, the dynamics of $Y_t$ is given by
\begin{equation*}
	\frac{\mathrm{d}Y_t}{Y_t}=\eta_t\mathrm{d}t+\mathbb{E}^1[\pi_t\sigma_t]\mathrm{d}W_t^0,\;\;\mathbb{P}^0\text{-a.s.},
\end{equation*}
where we define the auxiliary quantity
\begin{equation*}
	\eta_t=\mathbb{E}^1[\pi_t\mu_t]+\frac{1}{2}\big(\mathbb{E}^1[\pi_t\sigma_t]^2-\mathbb{E}^1[\pi_t^2\Sigma_t]\big).
\end{equation*}
We define a stochastic process $\alpha:\Omega\times[0,\infty)\to\mathcal{P}_2(\mathbb{R})$ by
\begin{equation*}
	\alpha_t(\omega^0)=\mathcal{L}\big(Y_t(\omega^0,\cdot)\big).
\end{equation*}

Note that the representative agent's utility is a random field $U:\Omega\times(0,\infty)\times\mathcal{P}_2(\mathbb{R})\times[0,\infty)\to\mathbb{R}$ such that,
\begin{equation}\label{eq:26}
	U(x,\mu,t)=\begin{cases}
		\big(1-\frac{1}{\delta_t}\big)^{-1}\Big(\lambda(\mu)^{-\theta_t}x\Big)^{1-\frac{1}{\delta_t}}K_t, & \text{if }\delta_t\neq1\\
		\log\Big(\lambda(\mu)^{-\theta_t}x\Big)K_t+G_t, & \text{if }\delta_t^i=1
	\end{cases},
\end{equation}
where $K_t$ and $G_t$ are $\mathcal{F}$-progressively measurable processes that are differentiable in time with $K_0=1$ and $G_0=0$.
Here, the random parameters satisfy the conditions $\delta_t=\delta(X_t,t)>0$, $\theta_t=\theta(X_t,t)\in[0,1]$, $\mathbb{P}$-almost surely, and they represent the $i$-th agent's relative risk tolerance and relative competition weight.\\

We refer the concept of the CRRA-MF-forward relative performance equilibrium and the CRRA-MF-equilibrium in \cite{RP3}, and extend these to the measure dependent forward relative performance process.

\begin{definition} (CRRA-MF-forward relative performance equilibrium for the generic agent)
	We assume that all agents have CRRA risk preference. The $\mathbb{F}$-progressively measurable random field $U:\Omega\times(0,\infty)\times\mathcal{P}_2(\mathbb{R})\times[0,\infty)\to\mathbb{R}$ is an MF-forward relative performance for the generic agent if, for all $t\geq0$, the following conditions hold:
	\begin{enumerate}
		\item $U(x,\mu,t)$ is strictly increasing and strictly concave in $x$, $\mathbb{P}$-a.s.
		\item For each $\pi_t\in\mathcal{A}_\text{MF}$, $U(X_t^\pi,\alpha_t^\pi,t)$ is a $\mathbb{P}^1$-(local) supermartingale, where $X_t^\pi$ is the wealth process solving (\ref{eq:25}) for the strategy $\pi_t$, and $\alpha_t^\pi$ is the corresponding conditional marginal distribution.
		\item There exists $\pi_t^*\in\mathcal{A}_\text{MF}$ such that $U(X_t^*,\alpha_t^*,t)$ is a $\mathbb{P}^1$-(local) martingale, where $X_t^*$ is the wealth process solving (\ref{eq:25}) for the strategy $\pi_t^*$, and $\alpha_t^*$ is the corresponding conditional marginal distribution.
	\end{enumerate}
	We call the optimal strategy $\pi_t^*$ as an MF-equilibrium, and denote the pair $(U,\pi^*)$ satisfying above conditions as an MF-forward relative performance equilibrium.
\end{definition}

We now present the fourth main finding, in which we show the existence of the MF-forward relative performance equilibrium for the generic agent when this agent has a CRRA risk preference.

\begin{theorem}\label{thm10}
	Define the random quantities
	\begin{equation*}
		\varphi^\sigma(t)=\mathbb{E}^1\Big[\frac{\mu_t\delta_t\sigma_t}{\Sigma_t}\Big],
	\end{equation*}
	\begin{equation*}
		\psi^\sigma(t)=\mathbb{E}^1\Big[(1-\delta_t)\frac{\theta_t(\sigma_t)^2}{\Sigma_t}\Big],
	\end{equation*}
	and let
	\begin{equation}\label{eq:27}
		K_t=\begin{cases}
			-\exp\bigg[\displaystyle\int_0^t\Big(1-\frac{1}{\delta_s}\Big)\bigg(-\theta_s\Big(\mathbb{E}^1[\pi_s\mu_s]-\frac{1}{2}\mathbb{E}^1[(\pi_s)^2\Sigma_s]\Big)&\\
			\qquad\qquad+\Big(1-\frac{1}{\delta_s}\Big)(\theta_s)^2\mathbb{E}^1[\pi_s\sigma_s]^2&\\
			\qquad\qquad+\frac{1}{2\delta_s\Sigma_s}\Big(\mu_s\delta_s+(1-\delta_s)\theta_s\sigma_s\mathbb{E}^1[\pi_s\sigma_s]\Big)^2\bigg)\mathrm{d}s\bigg], &\text{if } \delta_t\neq1\\
			\qquad1, &\text {if } \delta_t=1
		\end{cases},
	\end{equation}
	and
	\begin{equation}\label{eq:28}
		G_t=-\displaystyle\int_0^t\Big(-\theta_s\Big(\mathbb{E}^1[\pi_s\mu_s]-\frac{1}{2}\mathbb{E}^1[(\pi_s)^2\Sigma_s]\Big)+\frac{(\mu_s)^2}{2\Sigma_s}\Big)\mathrm{d}s,
	\end{equation}
	where
	\begin{equation*}
		\mathbb{E}^1[\pi_t\sigma_t]=\frac{\varphi^\sigma(t)}{1-\psi^\sigma(t)},
	\end{equation*}
	\begin{equation*}
		\mathbb{E}^1[\pi_t\mu_t]=\mathbb{E}^1\Big[\frac{(\mu_t)^2\delta_t}{\Sigma_t}\Big]+\mathbb{E}^1\Big[(1-\delta_t)\frac{\mu_t\theta_t\sigma_t}{\Sigma_t}\Big]\frac{\varphi^\sigma(t)}{1-\psi^\sigma(t)},
	\end{equation*}
	\begin{equation*}
		\mathbb{E}^1[(\pi_t)^2\Sigma_t]=\mathbb{E}^1\bigg[\frac{1}{\Sigma_t}\bigg(\mu_t\delta_t+(1-\delta_t)\theta_t\sigma_t\frac{\varphi^\sigma(t)}{1-\psi^\sigma(t)}\bigg)^2\;\bigg].
	\end{equation*}
	 Furthermore, we assume that $U(x,\mu,t)$, the utility of the representative agent, satisfies conditions (\ref{eq:2}), (\ref{eq:26}).
	If $\psi^\sigma(t)\neq 1$ $\mathbb{P}^0$-a.s. for every $t\geq0$, then there exists an optimal strategy $\pi_t^{*}$ defined by
	\begin{equation*}
		\pi_t^{*}=\frac{1}{\Sigma_t}\Big(\mu_t\delta_t+(1-\delta_t)\theta_t\sigma_t\frac{\varphi^\sigma(t)}{1-\psi^\sigma(t)}\Big),\;t\geq0.
	\end{equation*}
	Then $\big(U^{*},\pi^{*}\big)$ is an MF-forward relative performance equilibrium, where $U^{*}$ satisfies (\ref{eq:26}) with $X_t^{*},\;\alpha_t^{*},\;K_t$ and $G_t$.
\end{theorem}

\begin{proof}[Proof]
	As in the $n$-agent game, we use the Itô-Wentzell-Lions formula in Proposition \ref{prop2} to compute the dynamics of a random field $U(x,\mu,t)$. Then we have $\mathbb{P}^0$-a.s.,
	\begin{align}\label{eq:29}
		\begin{split}
			\mathrm{d}U(X_t,\alpha_t,t)=\;&\frac{\partial U}{\partial t}(X_t,\alpha_t,t)
			\mathrm{d}t+\frac{\partial U}{\partial x}
			(X_t,\alpha_t,t)\pi_tX_t\big(\mu_t\mathrm{d}t+\nu_t\mathrm{d}W_t+\sigma_t\mathrm{d}W_t^0\big)\\
			&+\frac{1}{2}\frac{\partial^2U}{\partial x^2}(X_t,\alpha_t,t)(\pi_t)^2(X_t)^2\Sigma_t\mathrm{d}t\\
			&+\widetilde{\mathbb{E}}^1\bigg[\frac{\partial U}{\partial \mu}
			(X_t,\alpha_t,t)(\widetilde{Y}_t)\frac{\widetilde{\eta}_t}{\widetilde{Y}_t}\bigg]\mathrm{d}t\\
			&+\widetilde{\mathbb{E}}^1\bigg[\frac{\partial U}{\partial \mu}
			(X_t,\alpha_t,t)(\widetilde{Y}_t)\frac{\widetilde{\mathbb{E}}^1[\widetilde{\pi}_t\widetilde{\sigma}_t]}{\widetilde{Y}_t}\bigg]\mathrm{d}W_t^0\\
			&+\frac{1}{2}\widetilde{\mathbb{E}}^1\bigg[\frac{\partial^2U}{\partial v \partial \mu}(X_t,\alpha_t,t)(\widetilde{Y}_t)\Big(\frac{\widetilde{\mathbb{E}}^1[\widetilde{\pi}_t\widetilde{\sigma}_t]}{\widetilde{Y}_t}\Big)^2\bigg]\mathrm{d}t\\
			&+\frac{1}{2}\widehat{\mathbb{E}}^1\bigg[\widetilde{\mathbb{E}}^1\bigg[\frac{\partial^2U}{\partial \mu^2}(X_t,\alpha_t,t)(\widetilde{Y}_t,\widehat{Y}_t)\frac{\widetilde{\mathbb{E}}^1[\widetilde{\pi}_t\widetilde{\sigma}_t]}{\widetilde{Y}_t}\frac{\widehat{\mathbb{E}}^1[\widehat{\pi}_t\widehat{\sigma}_t]}{\widehat{Y}_t}\bigg]\bigg]\mathrm{d}t\\
			&+\widetilde{\mathbb{E}}^1\bigg[\frac{\partial^2U}{\partial x \partial \mu}(X_t,\alpha_t,t)(\widetilde{Y}_t)\pi_tX_t\sigma_t\frac{\widetilde{\mathbb{E}}^1[\widetilde{\pi}_t\widetilde{\sigma}_t]}{\widetilde{Y}_t}\bigg]\mathrm{d}t.
		\end{split}
	\end{align}
	Substituting the derivatives in Proposition \ref{prop7} into (\ref{eq:29}), we obtain $\mathbb{P}^0$-a.s., in the case that $\delta_t\neq1$,
\begin{align}\label{eq:30}
	\begin{split}
		\frac{\mathrm{d}U(X_t,\alpha_t,t)}{U(X_t,\alpha_t,t)}=\;&\Big(1-\frac{1}{\delta_t}\Big)\bigg[-\frac{\Sigma_t}{2\delta_t}(\pi_t)^2+\frac{1}{\delta_t}\Big(\mu_t\delta_t+(1-\delta_t)\theta_t\sigma_t\mathbb{E}^1[{\pi}_t{\sigma}_t]\Big)\pi_t\\
		&\qquad\qquad \qquad+\Big(1-\frac{1}{\delta_t}\Big)(\theta_t)^2\mathbb{E}^1[{\pi}_t{\sigma}_t]^2 \\
		&\qquad\qquad \qquad-\theta_t\Big(\mathbb{E}^1[{\pi}_t{\mu}_t]-\mathbb{E}^1[(
		{\pi}_t)^2{\Sigma}_t]\Big)+\frac{1}{K_t}\frac{\mathrm{d}K_t}{\mathrm{d}t} \bigg]\mathrm{d}t\\
		&+\Big(1-\frac{1}{\delta_t}\Big)\bigg(\pi_t\nu_t\mathrm{d}W_t+\Big(\pi_t\sigma_t-\theta_t\mathbb{E}^1[{\pi}_t{\sigma}_t]\Big)\mathrm{d}W_t^0\bigg).
	\end{split}
\end{align}
	For the process $U(x,\mu,t)$ to be a forward relative performance, a value of the drift term in (\ref{eq:30}) should be less than or equal to 0 for all $\pi_t\in\mathcal{A}_\text{MF}$, and 0 for an optimal strategy. Note that the drift term in (\ref{eq:30}) is a quadratic function with respect to $\pi_t$ where the quadratic coefficient is negative.\\
	Let
	\begin{equation}\label{eq:31}
		\pi_t^*=\frac{1}{\Sigma_t}\Big(\mu_t\delta_t+(1-\delta_t)\theta_t\sigma_t\mathbb{E}^1[\pi_t\sigma_t]\Big),
	\end{equation}
	then the drift term is equal to $-\frac{1}{2}\frac{\Sigma_t}{\delta_t}(X_t)^{1-\frac{1}{\delta_t}}\vert \pi_t-\pi_t^{*}\vert^2$ when $K_t$ is given as in (\ref{eq:27}).\\
	Now multiplying both sides of (\ref{eq:31}) by $\sigma_t$, and taking conditional expectations on both sides, we obtain
	\begin{equation*}
		\mathbb{E}^1[\pi_t\sigma_t]=\frac{\varphi^\sigma(t)}{1-\psi^\sigma(t)}.
	\end{equation*}
	Substituting it into (\ref{eq:31}), we obtain a desired optimal strategy $\pi_t^{*}$.\\
	In the case that $\delta_t=1$, we can obtain analogous results, $\mathbb{P}^0$-a.s. we have,
	\begin{align}\label{eq:32}
	\begin{split}
		\mathrm{d}U(X_t,\alpha_t,t)=\;&\bigg[-\frac{\Sigma_t}{2}K_t(\pi_t)^2+\mu_t\pi_t-\theta_tK_t\Big(\mathbb{E}^1[{\pi}_t{\mu}_t]-\mathbb{E}^1[(
		{\pi}_t)^2{\Sigma}_t]\Big)\\
		&\qquad\qquad\qquad\qquad\qquad\quad+\Big(\log\Big(\lambda(\alpha_t)^{-\theta_t}X_t\Big)\frac{\mathrm{d}K_t}{\mathrm{d}t}+\frac{\mathrm{d}G_t}{\mathrm{d}t}\Big)\bigg]\mathrm{d}t\\
		&+K_t\pi_t\Big(\nu_t\mathrm{d}W_t+\sigma_t\mathrm{d}W_t^0\Big).
	\end{split}
\end{align}
	Thus, through a similar procedure, the drift term becomes $-\frac{1}{2}\Sigma_tK_t\vert \pi_t-\pi_t^*\vert^2$ when $K_t$ and $G_t$ are given as in (\ref{eq:27}), (\ref{eq:28}).\\
	For both cases, we can easily check that $\big(U^*,\pi^*\big)$ satisfies conditions given in Definition \ref{def7}. Thus, $\big(U^*,\pi^*)$ is an MF-forward relative performance equilibrium.\\
	Similarly, we have
	\begin{equation*}
		\mathbb{E}^1[\pi_t\mu_t]=\mathbb{E}^1\Big[\frac{(\mu_t)^2\delta_t}{\Sigma_t}\Big]+\mathbb{E}^1\Big[(1-\delta_t)\frac{\mu_t\theta_t\sigma_t}{\Sigma_t}\Big]\frac{\varphi^\sigma(t)}{1-\psi^\sigma(t)},
	\end{equation*}
	\begin{equation*}
		\mathbb{E}^1[(\pi_t)^2\Sigma_t]=\mathbb{E}^1\bigg[\frac{1}{\Sigma_t}\bigg(\mu_t\delta_t+(1-\delta_t)\theta_t\sigma_t\frac{\varphi^\sigma(t)}{1-\psi^\sigma(t)}\bigg)^2\;\bigg].
	\end{equation*}
\end{proof}

\begin{remark}
	As $n\to\infty$, the strategies $\pi_t^{i,*}$, weights $(\varphi_n^\sigma,\psi_n^\sigma)$, the utility $U^i(x,\mu,t)$ and $\mathcal{F}$-progressively measurable processes $K_t^i$ and $G_t^i$ in Theorem \ref{thm9} converge to the corresponding quantities in Theorem \ref{thm10}.
\end{remark}

\section{Conclusion}\label{sec5}
In this study, we show that the forward Nash equilibrium and the mean field equilibrium exist for the $n$-agent game and the corresponding mean field stochastic optimal control problem, respectively in the market model with random coefficients. Each agent looks for an optimal investment strategy that makes an utility a martingale. We focus on agents who have CARA or CRRA risk preferences for their investment optimization problem and consider five random coefficients which consist of three market parameters $\mu_t,\nu_t,\sigma_t$ and two preference parameters $\delta_t,\theta_t$. We conclude that our optimal portfolio formulas extend the corresponding findings of the market model with constant coefficients in \cite{RP1,RP3}. In our model with random coefficients, we can update the model coefficients in response to variations in market conditions in real-time.\\


\noindent\small{\textbf{Acknowledgement}} \footnotesize{The author wishes to thank his thesis supervisor Professor Geon Ho Choe for his encouragement over the years. He also thanks Professor Goncalo dos Reis for his helpful comments.}\\

\noindent\small{\textbf{Funding}} \footnotesize{The author is supported by the National Research Foundation of Korea (NRF) grant funded by the Korea government (MSIT) No. 2021R1A2C1010508.}

\end{document}